\documentclass[a4paper,onecolumn,11pt,accepted=2026-06-30]{quantumarticle}
\pdfoutput=1
\usepackage[utf8]{inputenc}
\usepackage[english]{babel}
\usepackage[T1]{fontenc}
\usepackage{amsmath,amsthm}
\usepackage{hyperref}

\usepackage{tikz}
\usepackage{lipsum}

\DeclareMathOperator{\id}{Id}
\DeclareMathOperator{\tr}{tr}
\usepackage{mathrsfs}
\newtheorem{remark}{Remark}[section]
\newtheorem{example}{Example}[section]
\newtheorem{proposition}{Proposition}[section]
\usepackage{float}
\usepackage{amsfonts}
\usepackage{algorithm,algorithmic}
\usepackage{multirow}
\def\b#1{\boldsymbol{#1}}

\begin{document}

\title{Reducing Spatial and Temporal Dimensionality in the Multidimensional Caldeira–Leggett Model}

\author[1]{Hongfei Zhan}
\orcid{0009-0003-4427-7538}
\email{hfzhan@nus.edu.sg}
\author[1]{Ernest W.Z. Pan}
\author[1]{Zhenning Cai}
\orcid{0000-0002-7086-7983}
\affil[1]{Department of Mathematics, National University of Singapore, Singapore 119076}
\maketitle

\begin{abstract}
Focusing on the real-time dynamics of the reduced density matrix of the multidimensional Caldeira-Leggett model, several techniques are adopted in this paper to reduce the spatial and temporal dimensionality, combined into an efficient algorithm.
From a spatial perspective, an equivalent formulation of the Dyson series is presented. With the aid of a low-rank approximation, the spatial dimensionality of open quantum system simulations is halved. 
From a temporal perspective, the frozen Gaussian approximation is used to approximate both the evolution operator and the interaction operator in the multidimensional Caldeira-Leggett model. This reduces the high-dimensional integrals to one- and two-dimensional integrals independent of the truncation level of the Dyson series. 
Through these techniques, we design an efficient algorithm whose validity is verified through several numerical experiments, including a two-dimensional double slit simulation.
\end{abstract}

\pdfoutput=1

\section{Introduction}\label{sec:intro}

An open quantum system refers to a quantum-mechanical system coupled to an external environment, with interactions that can significantly influence their quantum dynamics, leading to quantum effects such as quantum dissipation and quantum decoherence \cite{grigorescu1998decoherence,schlosshauer2019quantum}.
Due to its universality, the theory of open quantum systems has been widely used in diverse fields, encompassing quantum computing \cite{knill1997theory}, quantum communication \cite{nielsen2010quantum}, and quantum optical systems \cite{breuer2002theory}.

The open quantum system coupled to a harmonic bath is widely employed as a simplified model to study the dissipative effects of condensed-phase environments \cite{caldeira1983path} and also to simulate the effects of environments composed of normal mode vibrations, lattice phonons, and other more complex unstructured environments \cite{chandler1988introduction,makri1999linear}. A notable feature of the harmonic bath is that its degrees of freedom can be analytically integrated out in the path integral formulation, motivating the development of diverse numerical techniques for simulating open quantum systems.
A typical example is the spin-boson model \cite{chakravarty1984dynamics,thorwart2004dynamics,xu2023performance}, which describes a two-level system interacting with a bosonic environment. For open systems with more possible states, current research predominantly focuses on systems with special structures such as open spin chains \cite{makri2018modular,wang2023real,sun2024simulation}. In these systems, matrix product state representations are used to make the simulation tractable, but extending such approaches beyond one dimension remains challenging \cite{yan2021efficient,wang2023real,erpenbeck2023tensor,sun2024simulation}.
The spin-boson model can also be viewed as an implementation of the Caldeira-Leggett model \cite{caldeira1981influence,caldeira1983path,caldeira1985influence}, in which a particle is located in a double-well potential. 
However, most existing approaches remain limited to finite-dimensional systems, and numerical investigations of the Caldeira-Leggett model are confined to one-dimensional cases \cite{wang2025solving}.

Various models and methods have been developed to simulate open quantum systems. 
The Nakajima-Zwanzig equation \cite{nakajima1958quantum,zwanzig1960ensemble,shi2003new,zhang2006nonequilibrium,cohen2011memory,cohen2013numerically,wilner2013bistability,kelly2013efficient,kelly2016generalized} captures the temporal non-locality through a memory kernel. In the weak system-bath coupling regime, the process can be approximated by a Markovian process described by the Lindblad equation \cite{lindblad1976generators}. However, in more general cases, it would be difficult to bypass the non-Markovian nature during the numerical simulation. 
Due to the decay of the memory effect, a finite memory length prevents unlimited growth of storage or computational cost. The transfer tensor method (TTM) \cite{cerrillo2014non,rosenbach2016efficient} uses this technique to derive a discretized form of the Nakajima-Zwanzig equation. Alternatively, the hierarchical equations of motion (HEOM) \cite{tanimura1990nonperturbative,tanimura1989time,tanimura2020numerically} can be applied to the harmonic bath.
Recently, matrix product state and matrix product operator methods are widely used to efficiently simulate open quantum systems by exploiting their structure \cite{strathearn2018efficient,cygorek2022simulation,gribben2022exact,nunez2022learning,ng2023real,thoenniss2023efficient,nayak2025steady}.

The path integral \cite{feynman1948space,feynman1966quantum} plays a crucial role in the simulation of open quantum systems. It offers a classical-like picture of the quantum process. 
A variety of methods have been developed based on the path integral framework. 
Iterative path integral methods compute reduced dynamics by iteratively propagating the influence of environmental memory effects \cite{weiss2008iterative,segal2010numerically,simine2013path}.
The quasi-adiabatic propagator path integral (QuAPI) \cite{makri1992improved} uses the influence functional developed by Feynman and Vernon \cite{feynman1963theory} for non-Markovian dynamics.
Building upon QuAPI, many approaches have been proposed to improve simulation efficiency by reducing computational complexity or enhancing accuracy, including the iterative QuAPI method (i-QuAPI) \cite{makri1995numerical,makri1998quantum}, the blip decomposition of the path integral \cite{makri2014blip,makri2016blip}, the differential equation-based path integral method (DEBPI) \cite{wang2022differential}, and the kink sum method \cite{makri2024kink}.
Recently, the small matrix path integral (SMatPI) \cite{makri2020small-dynamics,makri2020small-memory,makri2020small-length,wang2024tree} successfully addressed the high memory demands associated with the summation over numerous paths by compactly representing their contributions through small matrices. Although the starting points are different, SMatPI yields a formulation akin to the Nakajima-Zwanzig equation.

Due to the large number of paths arising from the non-Markovian intrinsicality, a natural approach is the Monte Carlo methods. For example, the diagrammatic quantum Monte Carlo method (dQMC) uses diagrams to intuitively represent the coupling between the system and the bath \cite{prokof1998polaron,werner2009diagrammatic}. However, Monte Carlo methods often suffer from the numerical sign problem \cite{loh1990sign,cai2023numerical}. Different techniques, such as the application of bold lines \cite{prokof2007bold,prokof2008bold,chen2017inchworm,cai2020inchworm}, the inclusion-exclusion principle \cite{boag2018inclusion,yang2021inclusion}, and the combination of thin lines and bold lines \cite{cai2023bold}, have been developed to mitigate the numerical sign problem or accelerate the computation.

In this paper, an efficient algorithm is developed for real-time simulations of the reduced density matrix for the multidimensional Caldeira-Leggett model, with considerable effort devoted to reducing both the spatial and temporal dimensionality of the model. 

In the spatial domain, focusing on the harmonic bath, the open quantum system is depicted by the Dyson series, formulated in an equivalent form to the commonly used Keldysh contour. This representation allows for efficient algorithms through the reuse of intermediate variables. Moreover, a low-rank approximation is introduced for the two-point correlation function, facilitating the decomposition of diagrams associated with the Dyson series. Consequently, our formulation requires computation of the wave function rather than the density matrix, halving the spatial dimensionality. In addition, we also factorize numerous diagrams into the product of single diagrams, further reducing the computational cost. 

For the temporal dimensionality, separation between the interaction operator and the reduced density matrix is achieved through the frozen Gaussian approximation, which leads to a commutable scalar description of the interaction operator. The resulting multidimensional integrals in the expression of the reduced density matrix are further decomposed into one- and two-dimensional integrals. 
Building upon these techniques, we design an efficient algorithm to solve the multidimensional Caldeira-Leggett model. The temporal dependency of the computational cost is reduced to that of the first non-trivial term in the Dyson series. 
Numerical results, including the simulation of the two-dimensional double slit phenomenon, validate the effectiveness of the proposed method. To the best of our knowledge, this is the first algorithm capable of simulating the two-dimensional Caldeira-Leggett model.

The paper is organized as follows. In Sect. \ref{sec:oqs}, we introduce the open quantum system along with a low-rank approximation of the two-point correlation function. Next, the general-dimensional Caldeira-Leggett model is presented in Sect. \ref{subsec:cl model} and the frozen Gaussian approximation is introduced in Sect. \ref{subsec:fga}. We describe our main algorithm in Sect. \ref{subsec:alg}, with corresponding numerical results in Sect. \ref{sec:num result}. Lastly, we provide an overall conclusion in Sect. \ref{sec:conc}.
\section{Open quantum systems}\label{sec:oqs}
Consider the von Neumann equation for quantum evolution
\begin{equation}\label{eqn:von Neumann eqn}
	i\epsilon\frac{d\rho}{dt}
	=[H,\rho],
\end{equation}
where $\epsilon$ is the dimensionless parameter quantifying the ratio of quantum to classical scales, $\rho(t)$ is the density matrix at time $t$, and $H$ is the Schr\"odinger picture Hamiltonian in the form
\begin{align}
	H
	&=H_0+\epsilon W,\\
	H_0
	&=H_s\otimes\id_b+\id_s\otimes H_b,\\
	W
	&=W_s\otimes W_b.
\end{align}
Here $H_0$ is the Hamiltonian consisting of the system Hamiltonian $H_s$ and the bath Hamiltonian $H_b$, and $W$ depicts the interaction between the system $W_s$ and the bath $W_b$.
Assuming that the initial system is in the pure state $\psi_s^{(0)}$, with the bath in the thermal equilibrium state, we can express the initial density operator as
\begin{equation}
	\rho_0
	=\rho_s^{(0)}\otimes\rho_b^{(0)},
\end{equation}
where
\begin{align}
	\rho_s^{(0)}
	&:=|\psi_s^{(0)}\rangle\langle\psi_s^{(0)}|,\\
	\rho_b^{(0)}
	&:=\frac{e^{-\beta H_b }}{\tr_b\left(e^{-\beta H_b }\right)},
\end{align}
$\beta$ is the inverse temperature, and $\tr_b$ refers to the partial trace of the bath part. The solution of \eqref{eqn:von Neumann eqn} can be given by
\begin{equation}
	\rho(t)
	=e^{-iHt/\epsilon}\rho_0e^{iHt/\epsilon}.
\end{equation}
Then we focus on the reduced density matrix of the system as follows.
\begin{equation}\label{eqn:rho}
    \rho_s(t)
    =\tr_b\Big(\rho(t)\Big)
    =\tr_b\Big(e^{-iHt/\epsilon}\rho_0e^{iHt/\epsilon}\Big).
\end{equation}

In the following subsections, the Dyson series is introduced to approximate the reduced density matrix, while Wick's theorem is used to simplify the bath influence function. In particular, two equivalent forms of the bath influence function are presented, motivating the separation in Sect. \ref{subsec:equiv form}, and illustrating the possibility of further simplification in Sect. \ref{subsec:alg}. Consequently, a separated formulation of the reduced density matrix is derived with the aid of a low-rank approximation in Sect. \ref{subsec:lra}.

\subsection{Dyson series}\label{subsec:dyson}
Regarding the coupling term $\epsilon W$ as a perturbation of $H_0$, one can derive the Dyson series of the evolution operators as follows.
\begin{align}
    e^{-iHt/\epsilon}
    &=\sum\limits_{n=0}^\infty\int_{\b{s} \in \mathcal{S}_t^n}
        (-i)^ne^{-iH_0(t-s_n)/\epsilon}We^{-iH_0(s_n-s_{n-1})/\epsilon}W\cdots We^{-iH_0s_1/\epsilon}d\b{s},
\end{align}
where $\b{s}=(s_1,s_2,\dots,s_n)$, and the integral over an $n$-dimensional simplex of size $t$ is formulated as
\begin{equation}
    \int_{\b{s} \in \mathcal{S}_t^n} d\b{s}
	:=\int_0^{t}\int_0^{s_n}\cdots\int_0^{s_2}ds_1\cdots ds_{n-1}ds_n.
\end{equation}
Consequently, one can derive the Dyson series of the reduced density matrix in \eqref{eqn:rho} ,

\begin{align}
    \rho_{s,S}(t)
    &=e^{-iH_st/\epsilon}\rho_{s,I}(t)e^{iH_st/\epsilon},
    \label{eqn:dyson rho}\\
    \rho_{s,I}(t)
    &=\sum\limits_{n_1,n_2=0}^\infty
        \int_{\b{s}^{(1)} \in \mathcal{S}_t^{n_1}} \int_{\b{s}^{(2)} \in \mathcal{S}_t^{n_2}}
        \Big((-i)^{n_1} G_s(\b{s}^{(1)})\Big)
        |\psi_s^{(0)}\rangle\langle\psi_s^{(0)}|\notag\\
    &\qquad\qquad\qquad\qquad\qquad\qquad \cdot\Big((-i)^{n_2} G_s(\b{s}^{(2)})\Big)^\dagger
        \mathcal{L}_b(\b{s}^{(1)},\b{s}^{(2)})
            d\b{s}^{(1)}d\b{s}^{(2)},
    \label{eqn:dyson rho-interaction}
\end{align}
where $\rho_{s,S}(t), \rho_{s,I}(t)$ are the system density matrices in the Schr\"odinger picture and the interaction picture, respectively.

Also, $\b{s}^{(1)}=(s_1^{(1)}, \ldots, s_{n_1}^{(1)})$, $\b{s}^{(2)} = (s_1^{(2)}, \ldots, s_{n_2}^{(2)})$, and
\begin{equation}\label{eqn:evolution op}
	G_s(\b{s})
	=\mathcal{T}\prod\limits_{k=1}^nW_{s,I}(s_k),
	\qquad
	W_{s,I}(s_k)
	:=e^{iH_ss_k/\epsilon}W_se^{-iH_ss_k/\epsilon},
\end{equation} 
with $\mathcal{T}$ being the time-ordering operator. We also have the non-Markovian bath influence functional, defined by
\begin{equation}
\begin{aligned}
	\mathcal{L}_b(\b{s}^{(1)},\b{s}^{(2)})
	&=\tr\left(
		e^{-iH_b(t-s_{n_1}^{(1)})}W_b\cdots W_be^{-iH_bs_1^{(1)}}
		\rho_b^{(0)}
		e^{iH_bs_1^{(2)}}W_b\cdots W_be^{iH_b(t-s_{n_2}^{(2)})}
	\right).
\end{aligned}
\end{equation}
For the harmonic bath, by Wick's theorem \cite{wick1950evaluation}, it holds that
\begin{displaymath}
	\mathcal{L}_b(\b{s}^{(1)},\b{s}^{(2)})
	=\left\{\begin{array}{ll}
		0,	&\text{if $n_1+n_2$ is odd},\\
		\overline{\mathcal{L}}_b(s_1^{(1)}, \ldots, s_{n_1}^{(1)}, s_{n_2}^{(2)}, \ldots, s_1^{(2)}),
			&\text{if $n_1+n_2$ is even},
	\end{array}\right.
\end{displaymath}
where the function $\overline{\mathcal{L}}_b(\b{s})$ with a given $\b{s} = (s_1, \ldots, s_n)$ is defined by
\begin{equation} \label{eqn:bar Lb}
\overline{\mathcal{L}}_b(\b{s}) = \sum_{P \in \mathscr{P}_n} \prod_{(i,j) \in P} B(s_j, s_i).
\end{equation}
Here $\mathscr{P}_n$ contains all possible pairings of the integer set $\{1,\ldots,n\}$, and in each pair, the first element is smaller than the second. For instance,
\begin{align*}
  \mathscr{P}_2 = \big\{ &\{(1,2)\} \big\}, \\
  \mathscr{P}_4 = \big\{ &\{(1,2), (3,4)\}, \, \{(1,3), (2,4)\}, \, \{(1,4), (2,3)\} \big\}, \\
  \mathscr{P}_6 = \big\{ &\{(1,2), (3,4), (5,6)\}, \, \{(1,3), (2,4), (5,6)\}, \, \{(1,4), (2,3), (5,6)\}, \, \ldots, \, \\ &\{(1,6), (2,5), (3,4)\}\big\}.
\end{align*}
In general, for an even $n$, $\mathscr{P}_n$ contains $(n-1)!!$ elements.
The function $B(\cdot, \cdot)$ in \eqref{eqn:bar Lb} is known as the bath correlation function, which can be computed via
\begin{equation}
    B(\tau_1,\tau_2)
    =\tr\left(W_{b,I}(\tau_1)W_{b,I}(\tau_2)\rho_b^{(0)}\right),
    \qquad
    W_{b,I}(\tau)
    =e^{iH_b\tau/\epsilon}W_be^{-iH_b\tau/\epsilon}.
\end{equation}
Alternatively, $\overline{\mathcal{L}}_b(\b{s})$ can also be defined recursively:
\begin{displaymath}
\overline{\mathcal{L}}_b(\b{s}) = \sum_{k=2}^n B(s_k, s_1) \overline{\mathcal{L}}_b(\b{s} \backslash \{s_1, s_k\}), \qquad \overline{\mathcal{L}}_b(\emptyset) = 1.
\end{displaymath}
The definition of $B(\cdot, \cdot)$ is postponed to Sect. \ref{sec:num result}.

\begin{remark}
    \eqref{eqn:dyson rho-interaction} shows that a system initially in a pure state can evolve into a mixed state due to its interaction with the environment, hence we cannot use a wave function to describe an open quantum system.
\end{remark}

Following Feynman's idea, the Dyson series \eqref{eqn:dyson rho-interaction} can be understood as the sum of diagrams:
\begin{displaymath}
\rho_{s,I}(t) = \sum_{\substack{n_1,n_2 = 0\\ n_1 + n_2 \text{ is even}}}^{\infty} \sum_{P \in \mathscr{P}_{n_1+n_2}} \text{diag}(n_1,n_2,P),
\end{displaymath}
where the sum over $\mathscr{P}_{n_1+n_2}$ comes from the bath influence functional \eqref{eqn:bar Lb}.
Each diagram $\text{diag}(n_1, n_2, P)$ has the following structure:
\begin{enumerate}
\item The diagram consists of two axes, with the top axis representing the integral with respect to $\b{s}^{(1)}$ and the bottom axis representing the integral with respect to $\b{s}^{(2)}$. Any features on the bottom axis correspond to the conjugate transpose of the same feature on the top axis.
\item Both axes start with a red cross $\color{red} \b{\times}$ denoting $|\psi_s^{(0)}\rangle$ or $\langle\psi_s^{(0)}| = |\psi_s^{(0)}\rangle^\dagger$ in the first line of \eqref{eqn:dyson rho-interaction}. A dashed line connecting the red crosses shows the relative positions of $|\psi_s^{(0)}\rangle$ and $\langle\psi_s^{(0)}|$ being beside each other in \eqref{eqn:dyson rho-interaction}.
\item On the top axis, there are $n_1$ nodes corresponding to the operators $-i W_{s,I}(s_k^{(1)})$. Similarly, on the bottom axis, there are $n_2$ nodes corresponding to the operators $\left(-i W_{s,I}(s_k^{(2)})\right)^\dagger$. For both axes, the index $k$ increases from left to right. The proximity of the nodes to the cross on the same axis corresponds to the proximity of the respective operators to $|\psi_s^{(0)}\rangle \langle \psi_s^{(0)}|$ in \eqref{eqn:dyson rho-interaction}, with those on the top axis being to the left of $|\psi_s^{(0)}\rangle \langle \psi_s^{(0)}|$ and those on the bottom axis being to the right of $|\psi_s^{(0)}\rangle \langle \psi_s^{(0)}|$.
\item The nodes are paired using arcs, each representing a bath correlation function $B(\cdot,\cdot)$ whose parameters are the $s$-values represented by the nodes. These pairings are determined by $P$, and the position of $B(\cdot,\cdot)$ in \eqref{eqn:dyson rho-interaction} does not matter as it is a scalar. The input order follows Eq. \eqref{eqn:bar Lb}. Specifically, when both variables come from $\b{s}^{(1)}$, we get $B(s^{(1)}_j, s^{(1)}_i)$ where $i < j$. In contrast, when both variables come from $\b{s}^{(2)}$, the order is reversed due to taking the conjugate, $B(s^{(2)}_j, s^{(2)}_i)^* = B(s^{(2)}_i, s^{(2)}_j)$ where $i < j$. Otherwise, the variables associated with $\b{s}^{(2)}$ appear before those associated with $\b{s}^{(1)}$, e.g. $B(s^{(2)}_*, s^{(1)}_*)$.
\end{enumerate}

\begin{example}\label{exm:1d-comparison}
For $n_1 + n_2 = 2$, there are three diagrams that follow these rules, as shown below.
\begin{figure}[H]
    \centering
    \includegraphics[width=0.8\linewidth]{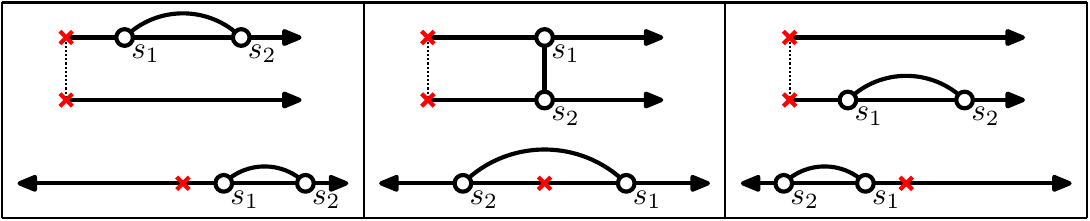}
    \caption{Comparison between diagrams (top) and corresponding Keldysh contour (bottom) for each term (from left to right) when $n_1+n_2=2$.}
    \label{fig:1d-comparison}
\end{figure}
This corresponds to the equation
\begin{align}
    &\int_{(s_1,s_2) \in \mathcal{S}_t^2}
    \Big( (-i)^2 W_{s,I}(s_2) W_{s,I}(s_1) \Big)
    |\psi_s^{(0)}\rangle\langle\psi_s^{(0)}|
    B(s_2, s_1)
    ds_1 ds_2 \notag\\
    &+ \int_{s_1 \in \mathcal{S}_t^1} \int_{s_2 \in \mathcal{S}_t^1}
    \Big( (-i) W_{s,I}(s_1) \Big)
    |\psi_s^{(0)}\rangle\langle\psi_s^{(0)}|
    \Big( i W_{s,I}(s_2)^\dagger \Big)
    B(s_2, s_1)
    ds_1 ds_2 \notag\\
    &+ \int_{(s_1,s_2) \in \mathcal{S}_t^2}
    |\psi_s^{(0)}\rangle\langle\psi_s^{(0)}|
    \Big( i^2 W_{s,I}(s_1)^\dagger W_{s,I}(s_2)^\dagger \Big)
    B(s_1, s_2)
    ds_1 ds_2,
    \label{eqn:eqn-to-diagram example}
\end{align}
where each integral matches the diagrams from left to right.

Consider the leftmost diagram of Fig. \ref{fig:1d-comparison}. Since there are two nodes on the top axis and none on the bottom axis, the integrals are over $\b{s}^{(1)}=(s_1,s_2) \in \mathcal{S}_t^2$ and $\b{s}^{(2)}=() \in \mathcal{S}_t^0$. On the top axis, $s_1$ is closer to the red cross than $s_2$, hence the described term is $\Big( (-i)W_{s,I}(s_2) (-i)W_{s,I}(s_1) \Big) |\psi_s^{(0)}\rangle\langle\psi_s^{(0)}|$. Lastly, the arc represents $B(s_2,s_1)$, and the overall diagram represents the first line of \eqref{eqn:eqn-to-diagram example}.

\end{example}

Thus, we can express the reduced density matrix by the sum of all possible diagrams satisfying the four rules given above:
\begin{equation} \label{eqn:rho diag}
    \includegraphics[width=1.\linewidth]{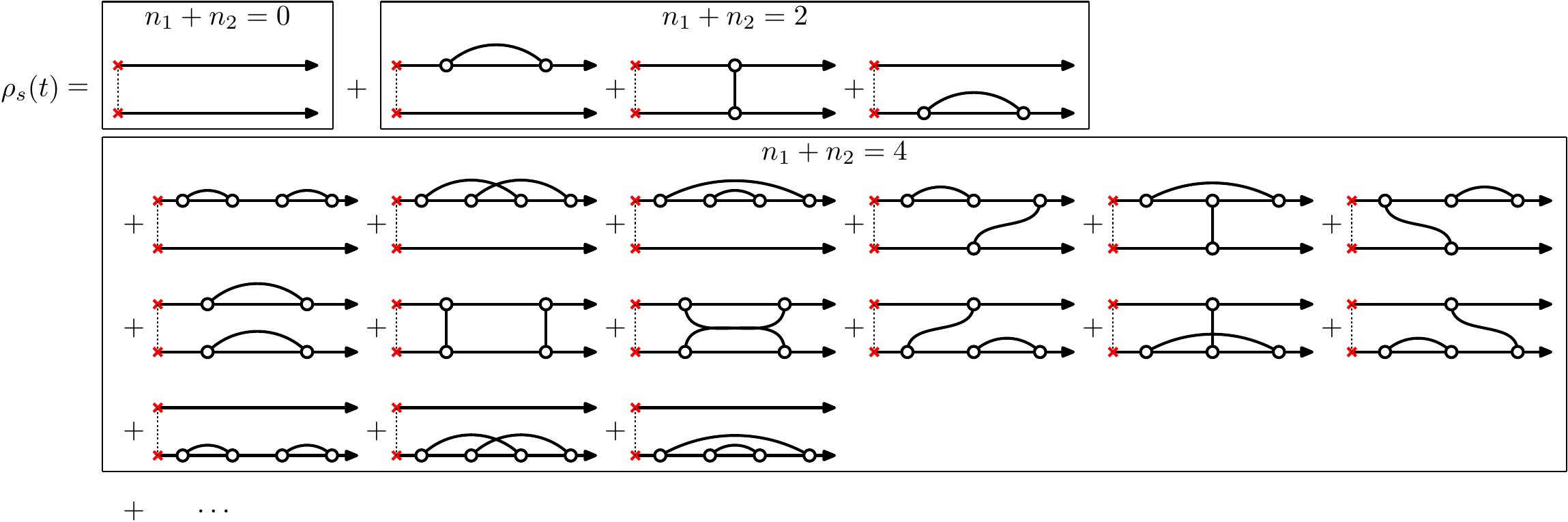}
\end{equation}
Such a formulation will provide a better understanding of the transformations to be performed in the following subsections, which may look tedious when written as equations. 
In particular, we aim to approximate the reduced density matrix as
\begin{equation} \label{eqn:rho fac}
    \rho_{s,I}
    \approx
    \sum\limits_{j=1}^Nc_j|\phi_j\rangle\langle\phi_j|,
\end{equation}
with weights $c_j$ and states $|\phi_j \rangle$ for a relatively small $N$. In contrast to the reduced density matrix, the evaluation of each diagram can be summed up by computing $|\phi_j\rangle$, halving the spatial dimensionality. To this end, we intend to separate the two axes in the diagrams, so that the upper and lower axes correspond to $|\phi_j\rangle$ and $\langle \phi_j|$, respectively. In Sect. \ref{subsec:equiv form}, the contribution on a single axis is first factored out. Subsequently, a low-rank approximation is introduced in Sect. \ref{subsec:lra} to complete the factorization such that $|\phi_j\rangle$ represents a sum of contributions on a single axis that share some feature.

\subsection{An equivalent form of reduced density matrix}\label{subsec:equiv form}
As the first step to achieve \eqref{eqn:rho fac}, we will separate arcs within a single axis from arcs across both axes. 
In \eqref{eqn:rho diag}, each diagram can be interpreted as a double integral with respect to the time sequences $\b{s}^{(1)}$ and $\b{s}^{(2)}$.
The objective of this subsection is to rewrite each diagram as a quadruple integral, whose corresponding four variables are:
\begin{enumerate}
\item Nodes on the top axis connecting to nodes on the bottom axis (black nodes on the top axis in Fig. \ref{fig:1d-decomposition}),
\item Nodes on the bottom axis connecting to nodes on the top axis (black nodes on the bottom axis in Fig. \ref{fig:1d-decomposition}),
\item Nodes on the top axis connecting to other nodes on the top axis (blue nodes in Fig. \ref{fig:1d-decomposition}),
\item Nodes on the bottom axis connecting to other nodes on the bottom axis (red nodes in Fig. \ref{fig:1d-decomposition}).
\end{enumerate}
\begin{figure}[H]
    \centering
    \includegraphics[width=0.4\linewidth]{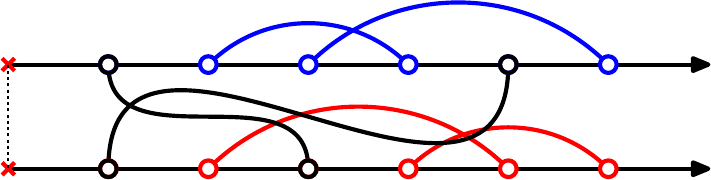}
    \caption{An example diagram illustrating the idea of decomposition of \eqref{eqn:dyson rho-interaction}.}
    \label{fig:1d-decomposition}
\end{figure}
Based on this idea, we can obtain the following reformulation of the reduced density matrix, in which blue and red are used to indicate the corresponding parts of Fig. \ref{fig:1d-decomposition}:
\begin{equation}\label{eqn:den mat-3 branch}
\begin{aligned}
	\rho_{s,I}(t)
	&=\sum\limits_{n=0}^\infty
            \int_{\b{s}^{(1)}, \b{s}^{(2)} \in \mathcal{S}_t^n}
            {
            \color{blue}\sum\limits_{m_1=0\atop\text{$m_1$ even}}^\infty
            (-1)^\frac{m_1}{2}
            \int_{\b{\tau}^{(1)} \in \mathcal{S}_t^{m_1}}
            }
            {
            \color{red}\sum\limits_{m_2=0\atop\text{$m_2$ even}}^\infty
            (-1)^\frac{m_2}{2}
            \int_{\b{\tau}^{(2)} \in \mathcal{S}_t^{m_2}}
            }\\
			&\qquad \mathcal{L}_b^{\rm cross}(\b{s}^{(1)},\b{s}^{(2)})
			{\color{blue} \mathcal{L}_b^{\rm same}(\b{\tau}^{(1)})}
            {\color{red}\Big(\mathcal{L}_b^{\rm same}(\b{\tau}^{(2)})\Big)^*} \\
			&\qquad 
            G_s\left([\b{s}^{(1)},\b{\tau}^{(1)}]\right)
			|\psi_s^{(0)}\rangle\langle\psi_s^{(0)}|
			G_s\left([\b{s}^{(2)},\b{\tau}^{(2)}]\right)^\dagger
		d\b{\tau}^{(1)}d\b{\tau}^{(2)}d\b{s}^{(1)}d\b{s}^{(2)},
\end{aligned}
\end{equation}
where $\b{s}^{(j)}=(s_1^{(j)}, \ldots, s_n^{(j)})$, $\b{\tau}^{(j)} = (\tau_1^{(j)}, \ldots, \tau_{m_j}^{(j)})$, and
\begin{align}
	\mathcal{L}_b^{\rm same}(\b{\tau})
	&= \overline{\mathcal{L}}_b(\b{\tau}), \\
	\mathcal{L}_b^{\rm cross}(\b{s}^{(1)},\b{s}^{(2)})
	&=\sum\limits_{\sigma\in\mathscr{Q}_n}
        \prod_{i=1}^n B(s_{\sigma(i)}^{(2)},s_i^{(1)}).
\end{align}
Note that the use of $[\cdot,\cdot]$ for the inputs of $G_s$ represents vector concatenation, and this convention will be used throughout this paper.
Here, $\mathscr{Q}_n$ is a set of $n!$ elements, including all possible permutations of $\mathcal{I}_n := \{1,\ldots,n\}$, represented as bijections from $\mathcal{I}_n$ to itself. For instance, $\mathscr{Q}_3$ contains $6$ elements $\sigma_1, \sigma_2, \sigma_3, \sigma_4, \sigma_5, \sigma_6$, given by
\begin{gather*}
\sigma_1(1) = 1, \qquad \sigma_2(1) = 1, \qquad \sigma_3(1) = 2, \qquad
\sigma_4(1) = 2, \qquad \sigma_5(1) = 3, \qquad \sigma_6(1) = 3, \\
\sigma_1(2) = 2, \qquad \sigma_2(2) = 3, \qquad \sigma_3(2) = 1, \qquad
\sigma_4(2) = 3, \qquad \sigma_5(2) = 1, \qquad \sigma_6(2) = 2, \\
\sigma_1(3) = 3, \qquad \sigma_2(3) = 2, \qquad \sigma_3(3) = 3, \qquad
\sigma_4(3) = 2, \qquad \sigma_5(3) = 2, \qquad \sigma_6(3) = 1.
\end{gather*}
Diagrammatically, $\mathcal{L}_b^{\rm same}(\b{\tau})$ sums up all possible arc configurations located on a single axis. In contrast, $\mathscr{L}_b^{\rm cross}(\b{s}^{(1)},\b{s}^{(2)})$ sums up all possible arc configurations across both axes, in the sense that each arc connects nodes on two different axes.

In the first line of equation \eqref{eqn:den mat-3 branch}, we have used different colors to label the sums and integrals corresponding to the arcs in Fig. \ref{fig:1d-decomposition}. This rearranged form still sums up all diagrams with every possible arc configuration within or across axes, which is equivalent to the definition given in \eqref{eqn:dyson rho-interaction}.

\begin{remark}
    Equivalently, $\mathcal{L}_b^{\rm same}$ and $\mathcal{L}_b^{\rm cross}$ can be expressed in the following recursive forms.
\begin{align}
	\mathcal{L}_b^{\rm same}(\b{\tau})
    &= \sum_{k=2}^{m} B(\tau_k,\tau_1) \mathcal{L}_b^{\rm same}(\b{\tau}\setminus\{\tau_1, \tau_k\}),\\
	\mathcal{L}_b^{\rm cross}(\b{s}^{(1)},\b{s}^{(2)})
	&=\left\{\begin{array}{ll}
            1,    &\text{if } n=0,\\
		\sum_{k=1}^{n} B(s_k^{(2)},s_1^{(1)})
        \mathcal{L}_b^{\rm cross}(\b{s}^{(1)}\setminus\{s_1^{(1)}\}, \b{s}^{(2)}\setminus\{s_k^{(2)}\}),
			&\text{if $n > 0$}.
	\end{array}\right.
\end{align}
Note that $\mathcal{L}_b^{\rm same}(\b{\tau}) = \mathcal{L}_b(\b{\tau})$.
This demonstrates the possibility to express the integral in \eqref{eqn:den mat-3 branch} as the product of one- and two-dimensional integrals when the contained operators commute.
\end{remark}

\begin{example}
Denote the possible pairings in $\mathcal{L}_b^{\rm same}(\b{\tau})$ as $\mathscr{L}_s(\b{\tau})$,
\begin{equation*}
\begin{aligned}
	\mathscr{L}_s(\b{\tau})
	=\left\{\begin{array}{ll}
		\Big\{\big\{(\tau_2,\tau_1)\big\}\Big\},
			&\text{if }\b{\tau}=(\tau_1,\tau_2),\\
		\Big\{\big\{(\tau_4,\tau_3),(\tau_2,\tau_1)\big\},\big\{(\tau_4,\tau_2),(\tau_3,\tau_1)\big\},\big\{(\tau_4,\tau_1),(\tau_3,\tau_2)\big\}\Big\},
			&\text{if }\b{\tau}=(\tau_1,\tau_2,\tau_3,\tau_4),\\
		\cdots,
			&
	\end{array}\right.
\end{aligned}
\end{equation*}
which corresponds to
\begin{equation*}
\begin{aligned}
	\mathcal{L}_b^{\rm same}(\tau_1,\tau_2)
	&=B(\tau_2,\tau_1),\\
	\mathcal{L}_b^{\rm same}(\tau_1,\tau_2,\tau_3,\tau_4)
	&=B(\tau_4,\tau_3)B(\tau_2,\tau_1)+B(\tau_4,\tau_2)B(\tau_3,\tau_1)+B(\tau_4,\tau_1)B(\tau_3,\tau_2),\\
	&\ \vdots
\end{aligned}
\end{equation*}
For instance, the three diagrams below correspond to the three terms in the sum for $\mathcal{L}_b^{\rm same}(\tau_1,\tau_2,\tau_3,\tau_4)$, in the same order.
\begin{figure}[H]
    \centering
    \includegraphics[width=0.9\linewidth]{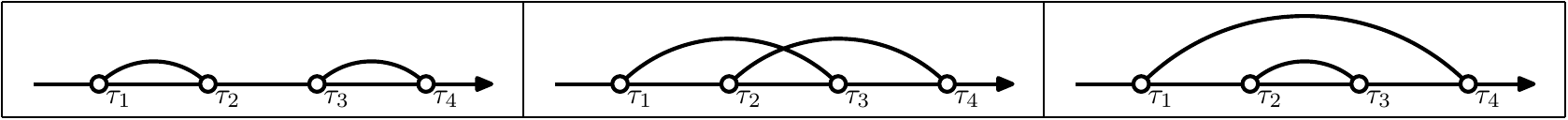}
    \caption{Example of integrals within a single axis in 1-d diagrams.}
    \label{fig:1d-same}
\end{figure}
The diagrams here represent the top axis, which differs from the bottom axis only by a conjugate transpose.
\end{example}

\begin{example}
Denote all possible pairings in $\mathcal{L}_b^{\rm cross}(\b{s}^{(1)},\b{s}^{(2)})$ as $\mathscr{L}_c(\b{s}^{(1)},\b{s}^{(2)})$,
\begin{align*}
	\mathscr{L}_c(\b{s}^{(1)},\b{s}^{(2)})
	&=\left\{\begin{array}{l}
		\Big\{\big\{(s_2,s_1)\big\}\Big\},\\
			\qquad\text{if }n=1,\b{s}^{(1)}=(s_1),\b{s}^{(2)}=(s_2),\\
		\Big\{\big\{(s_3,s_1),(s_4,s_2)\big\},
			\big\{(s_3,s_2),(s_4,s_1)\big\}\Big\},\\
			\qquad\text{if }n=2,\b{s}^{(1)}=(s_1,s_2),\b{s}^{(2)}=(s_3,s_4),\\
		\cdots,
	\end{array}\right.
\end{align*}
which corresponds to
\begin{align*}
	\mathcal{L}_b^{\rm cross}((s_1),(s_2))
	&=B(s_2,s_1),\\
	\mathcal{L}_b^{\rm cross}((s_1,s_2),(s_3,s_4))
	&=B(s_3,s_1)B(s_4,s_2) + B(s_3,s_2)B(s_4,s_1),\\
	&\ \vdots
\end{align*}
For instance, the two diagrams below correspond to the two terms in the sum for \\
$\mathcal{L}_b^{\rm cross}((s_1,s_2),(s_3,s_4))$, in the same order.
\begin{figure}[H]
    \centering
    \includegraphics[width=0.55\linewidth]{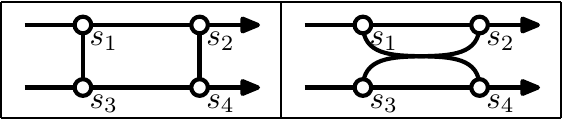}
    \caption{Example of integrals with arcs across both axes in 1-d diagrams.}
    \label{fig:1d-cross}
\end{figure}
\end{example}

\begin{remark}
	Compared to the Keldysh contour \cite{keldysh2024diagram} on $[-t,t]$, \eqref{eqn:den mat-3 branch} is an equivalent form that separates the contribution from $[-t,0]$, $[0,t]$ and those consisting of a two-point correlation crossing the origin point. This allows us to save computational cost by storing values of $G_s([\b{s}^{(1)},\b{\tau}^{(1)}])$ and reusing them when $G_s([\b{s}^{(2)},\b{\tau}^{(2)}])$ has matching inputs.
\end{remark}
In this subsection, we have split the integrals into three parts and identified the diagrams corresponding to each part. As a result, the integrands containing $\b{\tau}^{(1)}$ can be separated from the integrands containing $\b{\tau}^{(2)}$. In the next subsection, a low-rank approximation of the two-point correlation function is introduced to further separate the integrands containing $\b{s}^{(1)}$ from the integrands containing $\b{s}^{(2)}$. Consequently, we can factorize the integrals into independent parts.

\subsection{Low-rank approximation of two-point correlation function}\label{subsec:lra}
Factorization of the entire diagram remains a challenge due to the integration with respect to $\b{s}^{(1)}$ and $\b{s}^{(2)}$.
The purpose of this section is to separate the integrands containing these two variables so that they can be integrated independently, allowing for a complete factorization of the diagrams. Consequently, the evaluation of the reduced density matrix can be achieved by solely computing the wave function, halving the spatial dimensionality.
Since $\b{s}^{(1)}$ and $\b{s}^{(2)}$ are related by the function $\mathcal{L}_b^{\rm cross}(\cdot, \cdot)$, our plan is to apply low-rank approximation to this function, which is essentially a low-rank approximation of the bath correlation function $B(\cdot, \cdot)$.

We assume that the two-point correlation function, which satisfies $B(t_1, t_2) = [B(t_2, t_1)]^*$, has rank $r$, \textit{i.e.}, there exist scalars $\lambda_j$ and univariate functions $V_j(\cdot)$, $j = 1,\ldots,r$, such that
\begin{equation} \label{eq:low_rank_B}
	B(t_1,t_2)
	=\sum\limits_{j=1}^{r}\lambda_jV_j^*(t_1)V_j(t_2).
\end{equation}
Note that
\begin{equation*}
    \prod\limits_{k=1}^n V_{j_k}(s_{\sigma(k)})
    =\prod\limits_{k=1}^n V_{j_{\sigma^{-1}(k)}}(s_k),
\end{equation*}
for some permutation $\sigma \in \mathscr{Q}_n$.
Substituting the low-rank approximation into \eqref{eqn:den mat-3 branch}, we find that
\begin{equation*}
\begin{aligned}
    \rho_{s,I}(t)
    =&\sum\limits_{n=0}^\infty
        \sum_{\b{j}\in\{1,\cdots,r\}^n} \left(\prod\limits_{k=1}^n\lambda_{j_k}\right)
        \sum_{\sigma \in\mathscr{Q}_n}\\
        &\qquad\sum\limits_{m_1=0\atop\text{$m_1$ even}}^\infty(-1)^\frac{m_1}{2}
        \int_{\b{s}^{(1)} \in \mathcal{S}_t^n, \b{\tau}^{(1)} \in \mathcal{S}_t^{m_1}}
        \sum\limits_{m_2=0\atop\text{$m_2$ even}}^\infty(-1)^\frac{m_2}{2}
        \int_{\b{s}^{(2)} \in \mathcal{S}_t^n, \b{\tau}^{(2)} \in \mathcal{S}_t^{m_2}}\\
    &\qquad\left(\prod\limits_{k=1}^nV_{j_k}(s_k^{(1)})\right)
        \mathcal{L}_b^{\rm same}(\b{\tau}^{(1)})
        G_s\left([\b{s}^{(1)},\b{\tau}^{(1)}]\right)
	|\psi_s^{(0)}\rangle\langle\psi_s^{(0)}|\\
    &\qquad\cdot\left(\prod\limits_{k=1}^nV_{j_{\sigma(k)}}(s_k^{(2)})\right)^*
        \mathcal{L}_b^{\rm same}(\b{\tau}^{(2)})^*
        G_s\left([\b{s}^{(2)},\b{\tau}^{(2)}]\right)^\dagger
    d\b{\tau}^{(1)}d\b{\tau}^{(2)}d\b{s}^{(1)}d\b{s}^{(2)},
\end{aligned}
\end{equation*}
where $\b{j} = (j_1,\ldots,j_n)$, and the sum over $\b{j}$ comes from the product of $n$ bath correlation functions \eqref{eq:low_rank_B} in the expansion of $\mathcal{L}_b^{\rm cross}$.
The reduced density matrix can then be written as
\begin{align}
	\label{eqn:den mat-lra}
	\rho_{s,I}(t)
	&=\sum\limits_{n=0}^\infty
	\sum_{\b{j} \in \{1,\cdots,r\}^n}
	\left(\prod\limits_{k=1}^n\lambda_{j_k}\right)
	\sum\limits_{\sigma \in \mathscr{Q}_n}
		I_n(t, \b{j}) |\psi_s^{(0)}\rangle
		\Big(I_n(t, \b{j}_{\sigma}) |\psi_s^{(0)}\rangle\Big)^\dagger,\notag\\
    &\qquad \b{j}_{\sigma} = \left(j_{\sigma(1)}, \ldots, j_{\sigma(n)} \right), \\
	I_n(t, \b{j})
	&=
        \int_{\b{s} \in \mathcal{S}_t^n}
		\left(\prod\limits_{k=1}^nV_{j_k}(s_k)\right)
		\sum\limits_{m=0\atop\text{$m$ even}}^\infty(-1)^\frac{m}{2}
        \int_{\b{\tau} \in \mathcal{S}_t^m}
		\mathcal{L}_b^{\rm same}(\b{\tau})
        G_s\left([\b{s},\b{\tau}]\right)
		d\b{\tau}d\b{s}.
        \label{eqn:I-n,jk}
\end{align}
In practice, we usually apply the singular value decomposition to find $\lambda_j$ and $V_j$ such that the right-hand side of \eqref{eq:low_rank_B} approximates the function $B(\cdot,\cdot)$. 

The expression \eqref{eqn:den mat-lra} can also be represented by diagrams.
In general, $\rho_{s,I}(t)$ is still the sum of all diagrams with arbitrarily many arcs within or across axes.
However, for any arc across both axes, it is broken up into a sum of $r$ terms as in \eqref{eq:low_rank_B}.
As shown below, we will use different colors to denote different terms on the right-hand side of \eqref{eq:low_rank_B}.
For instance, when $r = 3$, we have 
\begin{equation} \label{fig:rho diagram lra}
    \includegraphics[width=1.\linewidth]{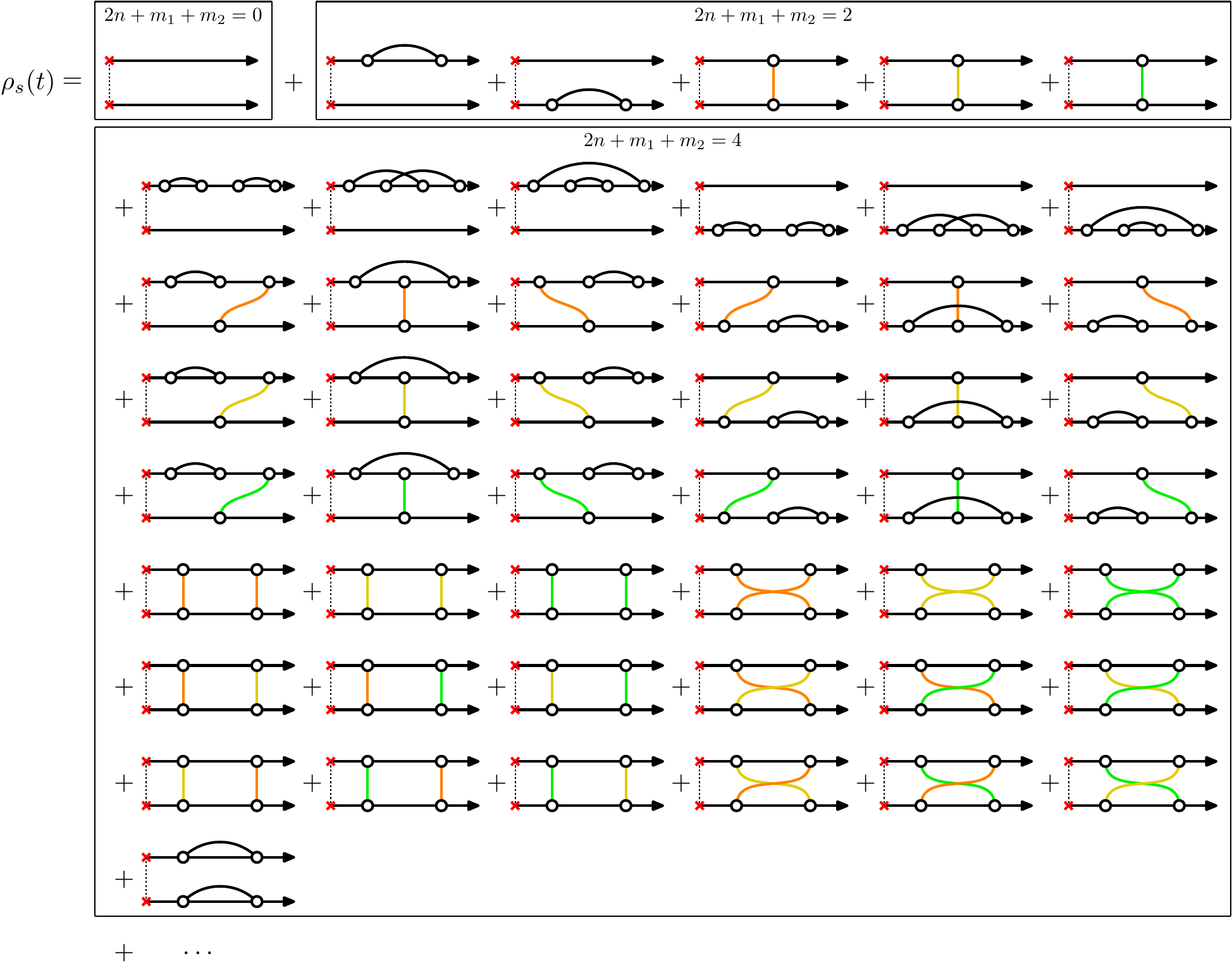}
\end{equation}
where the orange, yellow, and green arcs represent $\lambda_1 V_1^*(\cdot) V_1(\cdot)$, $\lambda_2 V_2^*(\cdot) V_2(\cdot)$, and $\lambda_3 V_3^*(\cdot) V_3(\cdot)$, respectively.
Furthermore, since each arc represents a product, we can separate the two axes in each diagram:
\begin{figure}[H]
    \centering
    \includegraphics[width=0.4\linewidth]{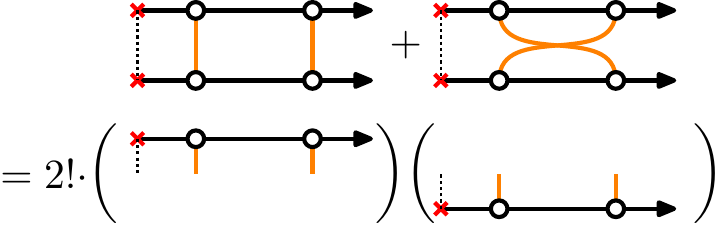}
    \caption{Diagram factorization for $n=2$, $\{j_1,j_2\}=\{1,1\}$.}
    \label{fig:factor example-same}

    \bigskip
    \includegraphics[width=0.8\linewidth]{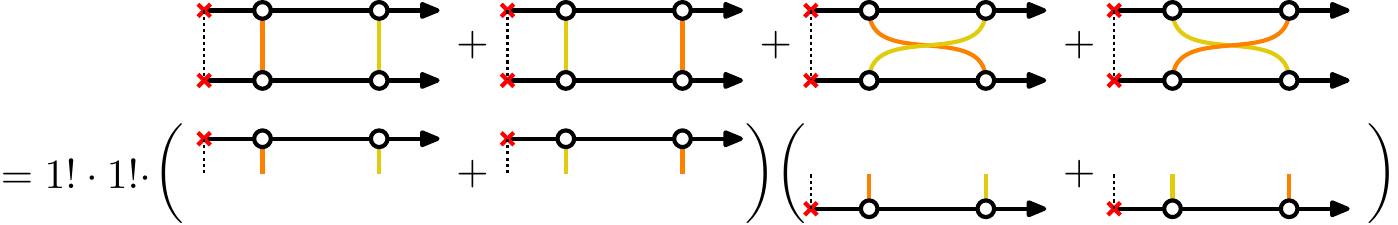}
    \caption{Diagram factorization for $n=2$, $\{j_1,j_2\}=\{1,2\},\{2,1\}$.}
    \label{fig:factor example-dif}
\end{figure}
\noindent so that terms with the same factor can be combined. In particular, the factor $2!$ in Fig. \ref{fig:factor example-same} and the factor $1!\cdot1!$ in Fig. \ref{fig:factor example-dif} represent the number of permutations for each colored arc. Each diagram in Fig. \ref{fig:factor example-same} has 2 arcs of the same color, hence $2!$ permutations, while each diagram in Fig. \ref{fig:factor example-dif} has 1 arc of two different colors, hence $1!\cdot1!$ permutations.

Further reduction of the formula requires a different representation of the sum over $j_1, \ldots, j_n$.
For each sequence $(j_1, \ldots, j_n)$, the quantity $N_i = \sum_{k=1}^n \delta_{i,j_k}$ represents the count of $i$ in the sequence.
Thus, given a multi-index  $\b{N}=(N_1,N_2,\ldots,N_r)$ satisfying $|\b{N}| := N_1 + \ldots + N_r = n$, the set
\begin{equation} \label{eqn:J N}
\mathcal{J}(\b{N}) = \left\{ \b{j} = (j_1,\ldots,j_n) \,\Bigg\vert\, \sum_{k=1}^n \delta_{i,j_k} = N_i \text{ for all } i = 1,\ldots,r \right\}
\end{equation}
contains all sequences whose index counts are given by components of $\b{N}$.
Therefore, the sum over $j_1, \ldots, j_n$ can be written as
\begin{displaymath}
\sum_{\b{j} \in \{1,\ldots,r\}^n} = \sum_{|\b{N}|=n} \sum_{\b{j} \in \mathcal{J}(\b{N})}.
\end{displaymath}
Note that for any $(j_1, \ldots, j_n) \in \mathcal{J}(\b{N})$, each of its permutations $(j_1', \ldots, j_n')$ is also an element in $\mathcal{J}(\b{N})$, and there exist $\b{N}!$ permutations $\sigma \in \mathscr{Q}_n$ such that $(j_1', \ldots, j_n') = (j_{\sigma(1)}, \ldots, j_{\sigma(n)})$, where $\b{N}! = N_1! \cdots N_r!$.
Therefore, the summation over $\sigma \in \mathscr{Q}_n$ in \eqref{eqn:den mat-lra} can be replaced by $\b{N}!$ times the summation over $\mathcal{J}(\b{N})$, yielding
\begin{displaymath}
\rho_{s,I}(t) = \sum_{n=0}^{\infty} \sum_{|\b{N}| = n} \b{\lambda}^{\b{N}} \b{N}! \sum_{\b{j} \in \mathcal{J}(\b{N})}\sum_{\b{j}' \in \mathcal{J}(\b{N})} 
    I_n(t, \b{j}) |\psi_s^{(0)}\rangle
    \Big(I_n(t, \b{j}') |\psi_s^{(0)}\rangle\Big)^\dagger,
\end{displaymath}
where $\b{\lambda}^{\b{N}} = \lambda_1^{N_1} \lambda_2^{N_2} \cdots \lambda_r^{N_r}$.
This factor $\b{N}!$ has been written explicitly in Fig. \ref{fig:factor example-same} and Fig. \ref{fig:factor example-dif}.
Finally, we obtain
\begin{align}
	\label{eqn:rho N}
	\rho_{s,I}(t)
	&=\sum\limits_{n=0}^\infty\sum\limits_{|\b{N}|=n}
		\b{\lambda}^{\b{N}}\b{N}!
		\Big(I_{n,\b{N}}(t)|\psi_s^{(0)}\rangle\Big)
		\Big(I_{n,\b{N}}(t)|\psi_s^{(0)}\rangle\Big)^\dagger,\\
	\label{eqn:integral N}
	I_{n,\b{N}}(t)
	&=
        \int_{\b{s} \in \mathcal{S}_t^n}
		\left(\sum\limits_{\b{j} \in \mathcal{J}(\b{N})}\prod\limits_{k=1}^nV_{j_k}(s_k)\right)
		\sum\limits_{m=0\atop\text{$m$ even}}^\infty(-1)^\frac{m}{2}
        \int_{\b{\tau} \in \mathcal{S}_t^m}
        \mathcal{L}_b^{\rm same}(\b{\tau})
        G_s\left([\b{s},\b{\tau}]\right)
		d\b{\tau}d\b{s}.
\end{align}
\begin{example} \label{exm:1d-LRA}
To illustrate the second summation in \eqref{eqn:rho N}, two different setups are presented in the following figures, showing all the distinct permutations.

For $r=3$, $n=2$, $\b{N}=(1,1,0)$, we get $\mathcal{J}(\b{N})=\{(1,2),(2,1)\}$.
\begin{figure}[H]
    \centering
    \includegraphics[width=0.45\linewidth]{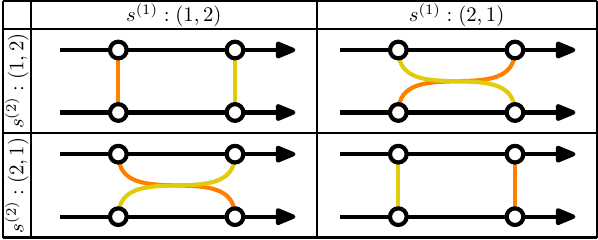}
    \caption{4 distinct diagrams for $\mathcal{J}(\b{N})=\{(1,2),(2,1)\}$.}
    \label{fig:LRA-110}
\end{figure}
For $r=5$, $n=3$, $\b{N}=(2,0,1,0,0)$, we get $\mathcal{J}(\b{N})=\{(1,1,3),(1,3,1),(3,1,1)\}$.
\begin{figure}[H]
    \centering
    \includegraphics[width=0.9\linewidth]{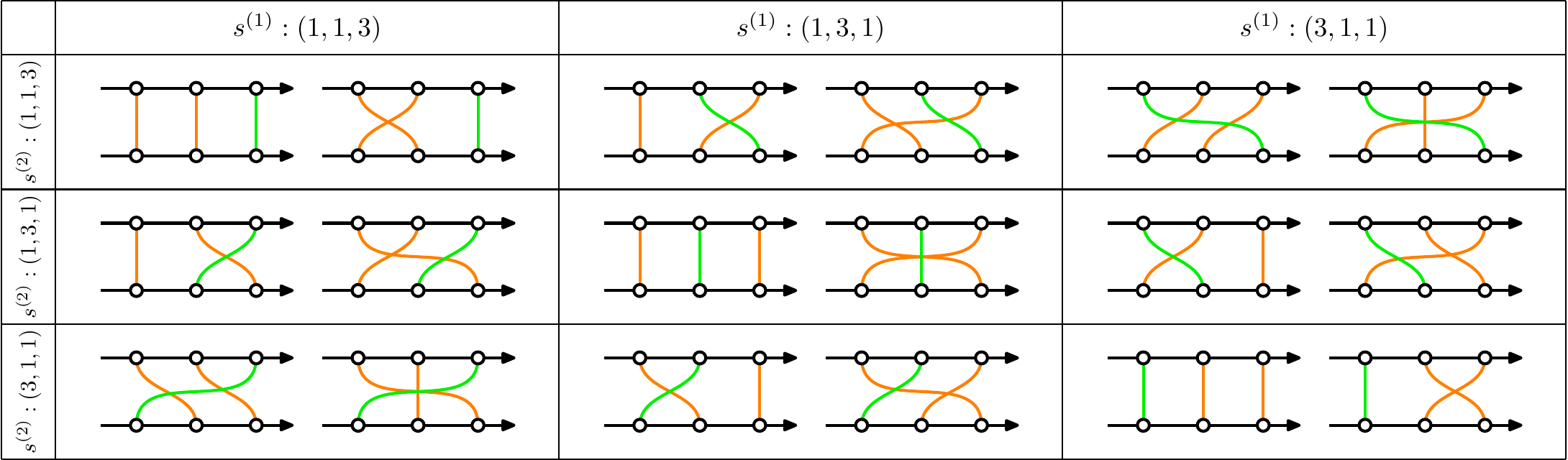}
    \caption{18 distinct diagrams for $\mathcal{J}(\b{N})=\{(1,1,3),(1,3,1),(3,1,1)\}$.}
    \label{fig:LRA-20100}
\end{figure}
In Figs. \ref{fig:LRA-110} and \ref{fig:LRA-20100}, different colored arcs represent different $V_j$ in the low-rank approximation and the color configurations on each axis are determined by the index set $\mathcal{J}(\b{N})$. Each Fig. only shows two distinct colors as there are only two nonzero elements in $\b{N}$ for both cases.

Consider the diagrams in row 2 and column 3 of Fig. \ref{fig:LRA-20100}. The orange and green arcs represent $V_1$ and $V_3$, respectively. In sequential order, the arc colors for nodes on the top axis are (green, orange, orange), corresponding to $\b{s}^{(1)}=(3,1,1)$, while the arc colors for nodes on the bottom axis are (orange, green, orange), corresponding to $\b{s}^{(2)}=(1,3,1)$. Specifically, $\b{s}^{(1)}=(3,1,1)$ means that for the top axis, the first node has an arc representing $V_3$, and the second and third nodes have arcs representing $V_1$. The two diagrams represent the only two configurations that satisfy the required arc colors.

Note that $I_{n,\b{N}}$ sums up $(n!)^2/\b{N}!$ diagrams with repeated features based on $\b{N}$. Although this summation does not reduce asymptotic space complexity, it offers substantial memory savings when $n$ is large.
\end{example}

\begin{remark}
    Compared to the single integral in \eqref{eqn:dyson rho-interaction} that corresponds to the contribution on a single axis, there exist two integrals in \eqref{eqn:integral N}. These arise from (i) the factorization of arcs on the single axis, (ii) the low-rank approximation to decompose the arcs across both axes. However, the number of temporal variables remains the same for these two formulas, and the latter expression allows a reduction in complexity.
\end{remark}

In this section, the factorization of the reduced density matrix in the interaction picture is achieved. In particular, instead of \eqref{eqn:rho fac}, we now have
\begin{equation}\label{eqn:general decomposition}
    c_j=\b{\lambda}^{\b{N}}\b{N}!,
    \qquad\text{and}\qquad
    |\phi_j\rangle=I_{n,\b{N}}(t)|\psi_s^{(0)}\rangle.
\end{equation}
As a result, the computation of the reduced density matrix requires only the calculation of $|\phi_j\rangle$, which halves the spatial dimensionality.

The current low-rank approximation framework shares several conceptual similarities with stochastic unraveling and quantum trajectory techniques \cite{diosi1997non,diosi1998non,suess2014hierarchy}. 
Compared to stochastic unraveling, which interprets the master equation as an average over individual quantum trajectories conditioned on environmental measurements, our method views the dynamics as a deterministic evolution within a truncated manifold. Both approaches exploit reduced wavefunction-based representations to alleviate the quadratic complexity associated with direct density-matrix propagation.

Although tremendous computational savings are achieved, the evaluation of $I_{n,\b{N}}(t)|\psi_s^{(0)}\rangle$ remains a difficulty due to the $n$-dimensional integral.
In the next section, we will focus on the Caldeira-Leggett model, where the frozen Gaussian approximation provides a fascinating feature that allows us to simplify this to a two-dimensional integral
regardless of the number of terms in the Dyson series.
\section{An efficient algorithm for Caldeira-Leggett model}
In this section, we apply the conclusion in Sect. \ref{sec:oqs} to the Caldeira-Leggett model, where a quantum particle is coupled to a harmonic bath.
Instead of direct discretization of the wave function, we apply the frozen Gaussian approximation that can effectively handle high-frequency waves in the solution. This also facilitates the design of an efficient algorithm for computing the density. In what follows, we will first review the Caldeira-Leggett model and the technique of frozen Gaussian approximation before introducing our numerical solver to compute the density.

\subsection{Caldeira-Leggett model}\label{subsec:cl model}
In most literature, the Caldeira-Leggett model is described as a one-dimensional particle coupled to a bath.
Here, we provide a multidimensional version of the model, which can be considered as several interacting one-dimensional particles coupled to separate heat baths.
This type of model has been studied extensively for many-particle open quantum systems \cite{makri2018modular,kundu2020modular,makri2024two}. 
In the $D$-dimensional Caldeira-Leggett model, the Hamiltonian can be formulated as
\begin{equation}
	H_s
	=-\frac{\epsilon^2}{2}\nabla_{\b{x}}^2+V(\b{x})+\sum\limits_{l=1}^L\frac{c_l^2}{2\omega_l^2}|\hat{\b{x}}|^2,
	\qquad
	H_b
	=\sum\limits_{l=1}^L\left(-\frac{\epsilon^2}{2}\nabla_{\b{z}_l}^2+\frac{1}{2}\omega_l^2|\hat{\b{z}_l}|^2\right),
\end{equation}
and the interaction operator describes the coupling between the particle and harmonic oscillators:
\begin{equation}\label{eqn:interaction-nd}
	W=\sum\limits_{d=1}^DW_s^{(d)}\otimes W_b^{(d)},
\end{equation}
where
\begin{equation}
	W_s^{(d)}
	=\hat{x}_d,
	\qquad\text{and}\qquad
	W_b^{(d)}
	=\frac{1}{\epsilon}\sum\limits_{l=1}^Lc_l\hat{z}_{l,d}.
\end{equation}
Here we provide a list of notations used in the equations above:
\begin{itemize}
	\item $\hat{\b{x}}$ ($\hat{x}_d $): The position operator of the particle in the system, $\psi(t,\b{x},\b{z})\mapsto\b{x}\psi(t,\b{x},\b{z})$ ($\psi(t,\b{x},\b{z})\mapsto x_d\psi(t,\b{x},\b{z})$).
	\item $\hat{\b{z}_l}$ ($\hat{z}_{l,d}$): The position operator of the particle in the bath, $\psi(t,\b{x},\b{z})\mapsto\b{z}_l\psi(t,\b{x},\b{z})$ ($\psi(t,\b{x},\b{z})\mapsto z_{l,d}\psi(t,\b{x},\b{z})$).
	\item $\omega_l$: The frequency of the $l $-th harmonic oscillator.
	\item $c_l$: The coupling intensity between the particle and the $l $-th harmonic oscillator.
	\item $V$: The potential function which is real and smooth.
\end{itemize}

The density matrix in the Caldeira-Leggett model can also be represented as a Dyson series like \eqref{eqn:dyson rho-interaction}.
However, when defining the propagator $G_s(\b{s})$ (see \eqref{eqn:evolution op}), the coupling operator $W_s$ now depends on the dimension $d$, leading to the definition
\begin{equation}\label{eqn:evolution op-nd}
    G_s(\b{s}, \b{d})
    =\mathcal{T}\prod\limits_{k=1}^nW_{s,I}^{(d_k)} (s_k),
    \quad
    W_{s,I}^{(d_k)}(s_k)
    :=e^{iH_ss_k}W_s^{(d_k)}e^{-iH_ss_k},
\end{equation}
where $\mathcal{T}$ is again the time-ordering operator that arranges the $W_{s,I}$ operators chronologically according to the time variable $s_k$, and $\b{d} = (d_1, \cdots, d_n)$ is a multi-index in $\{1,\cdots,D\}^n$.
The bath influence functional $\mathcal{L}_b$ will also depend on the dimension $d_k$.
Given $\b{s}^{(1)} \in \mathbb{R}_+^{n_1}$, $\b{d}^{(1)} \in \{1,\cdots,D\}^{n_1}$ and $\b{s}^{(2)} \in \mathbb{R}_+^{n_2}$, $\b{d}^{(2)} \in \{1,\cdots,D\}^{n_2}$, we have
\begin{equation}\label{eqn:bif-nd}
    \begin{aligned}
    & \mathcal{L}_b(\b{s}^{(1)}, \b{d}^{(1)};\b{s}^{(2)}, \b{d}^{(2)}) \\
    =& \left\{\begin{array}{ll}
        0,
            & \text{if $n_1+n_2$ is odd},\\
        \overline{\mathcal{L}}_b\Big((s_1^{(1)}, \ldots, s_{n_1}^{(1)},s_{n_2}^{(2)}, \ldots,s_1^{(2)}), (d_1^{(1)}, \ldots, d_{n_1}^{(1)},d_{n_2}^{(2)}, \ldots,d_1^{(2)})\Big),
            & \text{if $n_1+n_2$ is even},
    \end{array}\right.
    \end{aligned}
\end{equation}
where
\begin{equation} \label{eqn:Lb nd}
\overline{\mathcal{L}}_b\Big( (s_1, \ldots, s_n), (d_1, \ldots, d_n) \Big) = \sum_{P \in \mathscr{P}_n} \prod_{(i,j) \in P} \delta_{d_i, d_j} B(s_j, s_i).
\end{equation}
Thus, the Dyson series expansion of the density matrix turns out to be
\begin{equation} \label{eqn:dyson cl}
\begin{aligned}
	\rho_{s,I}(t)
	=\sum\limits_{n_1,n_2=0}^\infty
		&
        \int_{\b{s}^{(1)} \in \mathcal{S}_t^{n_1}}
        \sum_{\b{d}^{(1)} \in \{1,\cdots,D\}^{n_1}}
        \int_{\b{s}^{(2)} \in \mathcal{S}_t^{n_2}}
        \sum_{\b{d}^{(2)} \in \{1,\cdots,D\}^{n_2}}\\
		&\Big((-i)^{n_1} G_s(\b{s}^{(1)}, \b{d}^{(1)}) \Big) |\psi_s^{(0)}\rangle\langle\psi_s^{(0)}|\\
		&\cdot\Big((-i)^{n_2} G_s(\b{s}^{(2)}, \b{d}^{(2)}) \Big)^\dagger
        \mathcal{L}_b(\b{s}^{(1)}, \b{d}^{(1)}; \b{s}^{(2)},\b{d}^{(2)})
		d\b{s}^{(1)}d\b{s}^{(2)}.
\end{aligned}
\end{equation}

\begin{example}\label{exm:2d-comparison}
Due to the introduction of multiple dimensions, more complicated diagrams are needed to represent terms in the Dyson series \eqref{eqn:dyson cl}.
Note that each time point $s$ is bound to a dimension index $d$, and only when two time points are bound to the same dimension index, the two-point correlation function $B(\cdot,\cdot)$ has a contribution in the bath influence functional (see \eqref{eqn:Lb nd}).

This inspires us to draw $D$-dimensional diagrams to represent the integrands.
In each dimension, there are again two axes that accommodate points in $\b{s}^{(1)}$ and $\b{s}^{(2)}$, and due to the Kronecker symbol $\delta_{d_i,d_j}$ in \eqref{eqn:Lb nd}, only nodes located on axes of the same dimension can be connected, representing a bath correlation function.
One example is as follows:
\begin{figure}[H]
    \centering
    \includegraphics[width=0.8\linewidth]{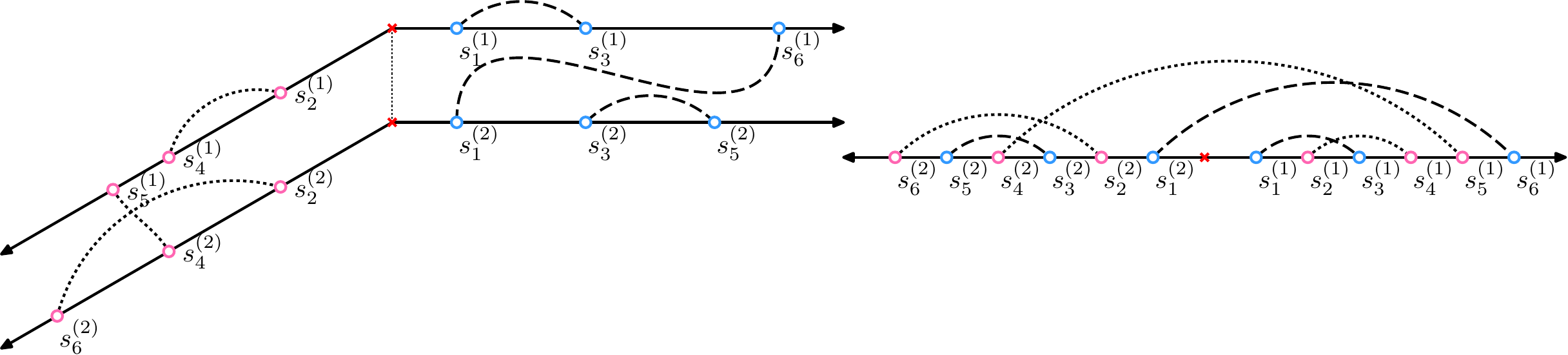}
    \caption{Comparison between diagram (left) and corresponding Keldysh contour (right).}
    \label{fig:2d-comparison}
\end{figure}

The left diagram in Fig. \ref{fig:2d-comparison} exhibits one sample in the case $D = 2$, while the right diagram is its ``flattened'' version. It represents the parameters $n_1 = n_2 = 6$, since both $\b{s}^{(1)}$ and $\b{s}^{(2)}$ have 6 components.
The corresponding $d$-values are
\begin{displaymath}
d_1^{(1)} = d_3^{(1)} = d_6^{(1)} = d_1^{(2)} = d_3^{(2)} = d_5^{(2)} = 1, \quad
d_2^{(1)} = d_4^{(1)} = d_5^{(1)} = d_2^{(2)} = d_4^{(2)} = d_6^{(2)} = 2,
\end{displaymath}
which can be observed by the location of each $s_k^{(1)}$ or $s_k^{(2)}$.
The connections between nodes show that this diagram corresponds to the term
\begin{displaymath}
B(s_3^{(1)}, s_1^{(1)}) B(s_1^{(2)}, s_6^{(1)}) B(s_3^{(2)}, s_5^{(2)}) B(s_4^{(1)}, s_2^{(1)}) B(s_4^{(2)}, s_5^{(1)}) B(s_2^{(2)}, s_6^{(2)})
\end{displaymath}
in the expansion of the bath influence functional \eqref{eqn:Lb nd}.
Meanwhile, we have also used blue and red nodes to denote the coupling operators $W_s^{(1)}$ and $W_s^{(2)}$, respectively.

In the flattened version, axes of different dimensions are combined, using different node colors and line styles for distinction: blue nodes with dashed arcs denote nodes in the first dimension and their connections, while red nodes with dotted arcs correspond to the second dimension.

The density matrix \eqref{eqn:dyson cl} can now be regarded as the sum of all such diagrams, where nodes can be located on any axis of any dimension, given that each dimension has an even number of nodes which are connected to other nodes within the same dimension.

\end{example}

We will now incorporate the low-rank decomposition of the bath correlation function \eqref{eq:low_rank_B} into the Dyson series.
Similar to Sect. \ref{subsec:lra}, the low-rank decomposition adds colors to arcs connecting the nodes, so that each dotted or dashed arc in Fig. \ref{fig:2d-comparison} can take any of the $r$ colors, each representing $\lambda_k V_k^*(\cdot) V_k(\cdot)$ for a certain $k = 1,\ldots,r$.
Using the same derivation as that in Sect. \ref{subsec:lra}, we obtain
\begin{equation} \label{eqn:rho N nd}
\begin{aligned}
\rho_{s,I}(t) = \sum_{n=0}^{\infty} \sum_{\substack{\b{N}^{(1)}, \ldots, \b{N}^{(D)} \in \mathbb{N}^r \\ |\b{N}^{(1)}| + \cdots + |\b{N}^{(D)}|=n}}
&\left( \prod_{d=1}^D \b{\lambda}^{\b{N}^{(d)}} \b{N}^{(d)}! \right)\\
&\cdot\Big( I_{n,\b{N}^{(1)}, \ldots, \b{N}^{(D)}}(t) |\psi_s^{(0)}\rangle \Big)
\Big( I_{n,\b{N}^{(1)}, \ldots, \b{N}^{(D)}}(t) |\psi_s^{(0)}\rangle \Big)^\dagger.
\end{aligned}
\end{equation}
Before defining the operator $I_{n,\b{N}^{(1)}, \ldots, \b{N}^{(D)}}$, we note that the multi-index $\b{N}^{(d)}$ denotes the numbers of arcs of different colors connecting nodes on axes in the $d$-th dimension, so that \eqref{eqn:rho N nd} is a natural extension of \eqref{eqn:rho N} to $D$ dimensions.
The generalization of $I_{n,\b{N}}$ defined in \eqref{eqn:integral N} needs the introduction of the index set
\begin{equation} \label{eqn:JnN}
\mathcal{J}(\b{N}^{(1)}, \ldots, \b{N}^{(D)}) := \left\{ (\b{j}, \b{d}) \,\Bigg|\, \sum_{k=1}^n \delta_{j_k,i} \delta_{d_k,d} = N_i^{(d)}, \, \forall i=1,\ldots,r, \, d = 1,\ldots,D\right\},
\end{equation}
which means that for any $(\b{j}, \b{d}) \in \mathcal{J}(\b{N}^{(1)}, \ldots, \b{N}^{(D)})$, we can find $N_i^{(d)}$ pairs of $(j_k, d_k)$ equal to $(i,d)$.
When $D = 1$, it reduces to the previous definition \eqref{eqn:J N} since $d_k \equiv 1$.
The $D$-dimensional generalization of \eqref{eqn:integral N} can then be expressed as
\begin{equation}\label{eqn:integral N-nd}
\begin{aligned}
I_{n,\b{N}^{(1)}, \ldots, \b{N}^{(D)}}(t)
	&=\int_{\b{s} \in \mathcal{S}_t^n}
		\sum\limits_{(\b{j},\b{d})\in\mathcal{J}(\b{N}^{(1)}, \cdots, \b{N}^{(D)})} \left(\prod\limits_{k=1}^nV_{j_k}(s_k)\right) \\
		&\qquad\sum\limits_{m=0\atop\text{$m$ even}}^\infty(-1)^\frac{m}{2}
        \int_{\b{\tau} \in \mathcal{S}_t^m}
			\sum_{\b{\kappa}\in \{1,\cdots,D\}^m} \mathcal{L}_b^{\rm same}(\b{\tau}, \b{\kappa})
			G_s\left([\b{s},\b{\tau}], [\b{d},\b{\kappa}]\right)
		d\b{\tau}d\b{s},
\end{aligned}
\end{equation}
where
\begin{equation}
    \mathcal{L}_b^{\rm same}(\b{\tau}, \b{\kappa}) = \overline{\mathcal{L}}_b(\b{\tau}, \b{\kappa}).
\end{equation}

For simplicity, we will use $\b{N} = (\b{N}^{(1)}, \ldots, \b{N}^{(D)})$ to denote the collection of the $D$ multi-indices, so that $I_{n,\b{N}^{(1)}, \ldots, \b{N}^{(D)}}$ and $\mathcal{J}(\b{N}^{(1)}, \ldots, \b{N}^{(D)})$ can be simplified to $I_{n,\b{N}}$ and $\mathcal{J}(\b{N})$, resembling the notations in the one-dimensional case.
Furthermore, if we let
\begin{equation*}
    |\b{N}|
	:=\sum_{d=1}^D |\b{N}^{(d)}|,
	\qquad
	\b{N}!
	:=\prod_{d=1}^D \b{N}^{(d)}!,
	\qquad
	\b{\lambda}^{\b{N}}
	:=\prod_{d=1}^D \b{\lambda}^{\b{N}^{(d)}},
\end{equation*}
then the reduced density matrix is again given by \eqref{eqn:rho N}, while $I_{n,\b{N}}(t)$ should adopt the definition in \eqref{eqn:integral N-nd}.
These simplified notations will be utilized hereafter.

\begin{example}
As an extension to Example \ref{exm:1d-LRA}, here we present three different setups and show all the distinct permutations.

For  
$D=1$, $r=3$, $n=2$, $\b{N}^{(1)} = (1,1,0)$, 
we get $\mathcal{J}(\b{N}) = \{\{(1,1),(2,1)\},\{(2,1),(1,1)\}\}$. 
This case corresponds to the first case in Example \ref{exm:1d-LRA}.

For  
$D=2$, $r=3$, $n=2$, $\b{N}^{(1)} = (0,1,0)$, $\b{N}^{(2)} = (0,0,1)$, 
we get $J_2 = \mathcal{J}(\b{N}) = \{\{(2,1),(3,2)\}$, $\{(3,2),(2,1)\}\}$.
\begin{figure}[H]
    \centering
    \includegraphics[width=0.95\linewidth]{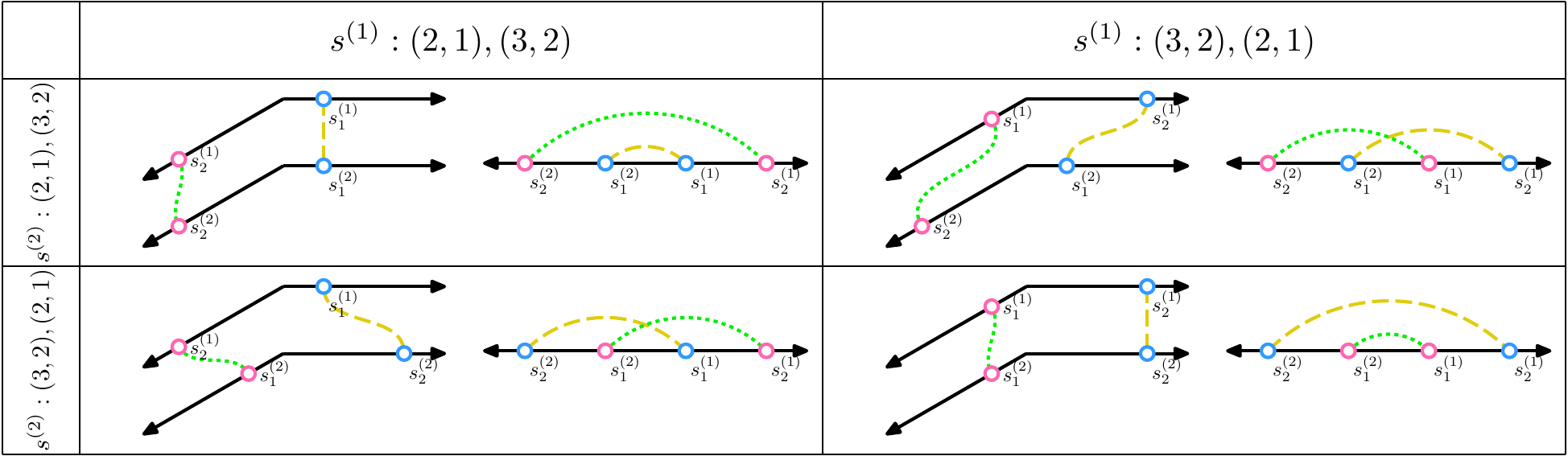}
    \caption{4 distinct diagrams and corresponding Keldysh contours for $J_2$.}
    \label{fig:2d-LRA-010001}
\end{figure}
In Fig. \ref{fig:2d-LRA-010001}, with the aid of Keldysh contours, one can obtain a similar result as the first case in Example \ref{exm:1d-LRA}. Consider the diagrams in row 1 and column 2. On the top axis, (i) the first node $s_1^{(1)}$ has the arc representing $V_3$ and corresponds to the interaction operator in dimension 2; (ii) the second node $s_2^{(1)}$ has the arc representing $V_2$ and corresponds to the interaction operator in dimension 1.

For $D=2$, $r=3$, $n=3$, 
$\b{N}^{(1)}=(2,0,0)$, $\b{N}^{(2)}=(0,0,1)$,
we get \\ $J_3 = \mathcal{J}(\b{N})=\{\{(1,1),(1,1),(3,2)\},\{(1,1),(3,2),(1,1)\},\{(3,2),(1,1),(1,1)\}\}$.
\begin{figure}[H]
    \centering
    \includegraphics[width=0.95\linewidth]{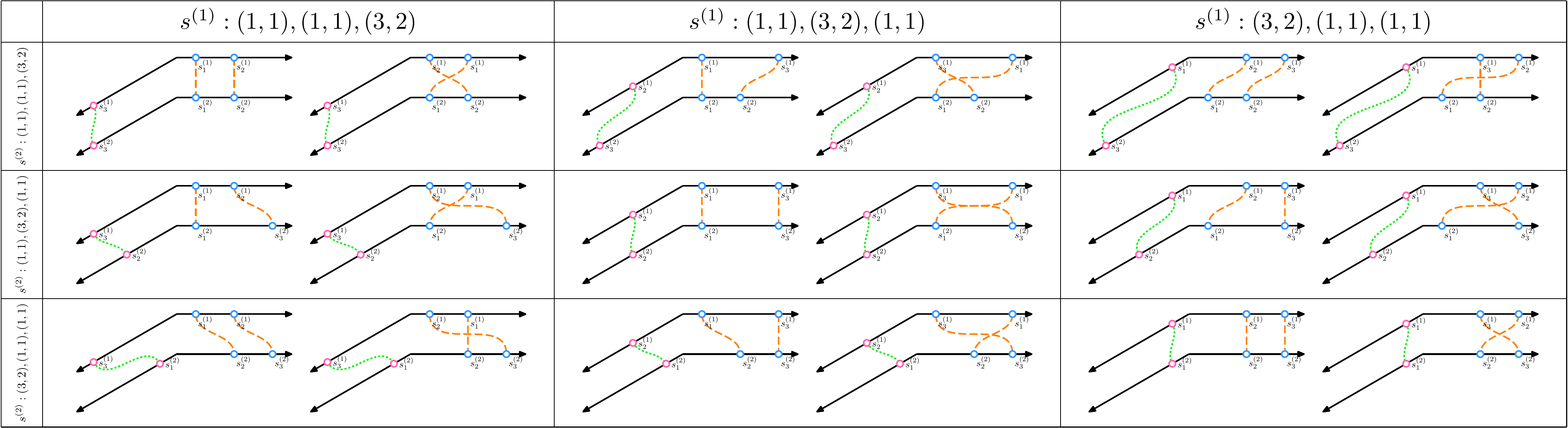}
    \caption{18 distinct diagrams for $J_3$.}
    \label{fig:2d-LRA-200001}
\end{figure}
\end{example}

Note that $I_{n,\b{N}}(t)$ is difficult to compute due to the time-ordering operator. In the next subsection, the frozen Gaussian approximation is adopted for a simple description of $G_s(\cdot)$, which provides the basis for developing an efficient algorithm to compute the reduced density matrix.

\subsection{Frozen Gaussian Approximation}\label{subsec:fga}

Although much simplification is considered in the former section, direct computations of $D$-dimensional $|\phi_j\rangle$ in \eqref{eqn:general decomposition} remain expensive. In this subsection, the Frozen Gaussian Approximation (FGA) is introduced to discretize the system wave function. This validates a simple approximation of the interaction operator, as well as acceleration by parallel implementation.

FGA \cite{heller1981frozen,lu2011frozen,wang2025solving} utilizes the ansatz
\begin{equation}\label{eqn:fga ansatz}
	\psi_{\rm FGA}(t,\b{x})
	=\frac{1}{(2\pi\epsilon)^{3D/2}}\iint_{\b{p},\b{q}\in\mathbb{R}^D}\int_{\b{y}\in\mathbb{R}^D}
		a(t,\b{p},\b{q})e^{i\phi(t,\b{x},\b{y},\b{p},\b{q})/\epsilon}\psi_s^{(0)}(\b{y})d\b{y}d\b{p}d\b{q},
\end{equation}
where
\begin{equation}
	\phi(t,\b{x},\b{y},\b{p},\b{q})
	=S(t,\b{p},\b{q})+\frac{i}{2}|\b{x}-\b{Q}|^2+\b{P}\cdot(\b{x}-\b{Q})
		+\frac{i}{2}|\b{y}-\b{q}|^2-\b{p}\cdot(\b{y}-\b{q}).
\end{equation}
Then
\begin{equation}
    \psi_{\rm FGA}(0,\b{x})
    =\psi_s^{(0)}(\b{x}),
    \qquad
    \psi_{\rm FGA}(t,\b{x})
    =e^{-iH_st/\epsilon}\psi_s^{(0)}(\b{x})+O(\epsilon).
\end{equation}
For brevity, we take
\begin{equation*}
	\tilde{V}(\b{x})
	=V(\b{x})+\sum\limits_{l=1}^L\frac{c_l^2}{2\omega_l^2}|\b{x}|^2.
\end{equation*}
As functions of $(t,\b{p},\b{q})$, the variables $\b{P}$, $\b{Q}$, $S$, $a$ satisfy the following dynamics:
\begin{equation}\label{eqn:FGA-ODEs}
\begin{aligned}
	\frac{\partial\b{P}}{\partial t}
	&=-\nabla\tilde{V}(\b{Q}), \qquad
	&&\frac{\partial\b{Q}}{\partial t}
	=\b{P},\\
	\frac{\partial S}{\partial t}
	&=\frac{|\b{P}|^2}{2}-\tilde{V}(\b{Q}), \qquad
	&&\frac{\partial a}{\partial t}
	=\frac{a}{2}\tr\left(
		\b{Z}^{-1}\left(
			\nabla_{\b{z}}\b{P}
			-i\nabla_{\b{z}}\b{Q}\nabla^2\tilde{V}(\b{Q})\right)
	\right),
\end{aligned}
\end{equation}
where
\begin{equation*}
	\nabla_{\b{z}}
	=\nabla_{\b{q}}-i\nabla_{\b{p}},
	\qquad\text{and}\qquad
	\b{Z}
	=\nabla_{\b{z}}\b{Q} + i\nabla_{\b{z}}\b{P}.
\end{equation*}
The initial conditions are
\begin{equation}
	\b{P}(0,\b{p},\b{q})
	=\b{p},
	\qquad
	\b{Q}(0,\b{p},\b{q})
	=\b{q},
	\qquad
	S(0,\b{p},\b{q})
	=0,\\
	\qquad
	a(0,\b{p},\b{q})
	=2^{D/2}.
\end{equation}
To solve this differential equation system, we also need the equations
\begin{equation}
	\frac{\partial(\nabla_{\b{z}}\b{P})}{\partial t}(t,\b{p},\b{q})
	=-\nabla_{\b{z}}\b{Q}(t,\b{p},\b{q})\nabla^2\tilde{V}(\b{Q}(t,\b{p},\b{q})),
	\qquad
	\frac{\partial(\nabla_{\b{z}}\b{Q})}{\partial t}(t,\b{p},\b{q})
	=\nabla_{\b{z}}\b{P}(t,\b{p},\b{q}),
\end{equation}
with initial conditions
\begin{equation}
	\nabla_{\b{z}}\b{P}(0,\b{p},\b{q})
	=-iI_D,
	\qquad
	\nabla_{\b{z}}\b{Q}(0,\b{p},\b{q})
	=I_D.
\end{equation}

It follows \cite{lu2011frozen,wang2025solving} that FGA provides the following approximations,
\begin{equation}
    W_{s,I}^{(d)}(s)
    =Q_d(s,\b{p},\b{q})+O(\epsilon),
\end{equation}
and
\begin{equation}
\begin{aligned}
    &e^{-iH_s(t-t_0)/\epsilon}\bigg{(}
        \frac{1}{(2\pi\epsilon)^{3D/2}}\iint_{\b{p},\b{q}\in\mathbb{R}^{D}}\int_{\b{y}\in\mathbb{R}^D}
        f(\b{p},\b{q})   a(t_0,\b{p},\b{q})e^{i\phi(t_0,\b{x},\b{y},\b{p},\b{q})/\epsilon}\psi_s^{(0)}(\b{y})
        d\b{y}d\b{p}\b{q}
        \bigg{)}\\
    =&
        \frac{1}{(2\pi\epsilon)^{3D/2}}\iint_{\b{p},\b{q}\in\mathbb{R}^{D}}\int_{\b{y}\in\mathbb{R}^D}
        f(\b{p},\b{q})        a(t,\b{p},\b{q})e^{i\phi(t,\b{x},\b{y},\b{p},\b{q})/\epsilon}\psi_s^{(0)}(\b{y})
        d\b{y}d\b{p}\b{q},
\end{aligned}
\end{equation}
for an arbitrary factor $f(\b{p},\b{q})$. Hence,
\begin{equation}
\begin{aligned}
    \rho_{s}(t,\b{x})
    &=\langle\b{x}|e^{-iH_st/\epsilon}\rho_{s,I}(t)e^{iH_st/\epsilon}|\b{x}\rangle\\
    &=\sum\limits_{n=0}^\infty\sum\limits_{|\b{N}|=n}
		\b{\lambda}^{\b{N}}\b{N}!
		\Big(\langle\b{x}|e^{-iH_st/\epsilon}I_{n,\b{N}}(t)|\psi_s^{(0)}\rangle\Big)
		\Big(\langle\b{x}|e^{-iH_st/\epsilon}I_{n,\b{N}}(t)|\psi_s^{(0)}\rangle\Big)^*,
\end{aligned}
\end{equation}
where
\begin{equation}\label{eqn:InN(t,x)-FGA}
\begin{aligned}
	&I_{n,\b{N}}(t,\b{x})
	:=\langle \b{x}|e^{-iH_st/\epsilon}I_{n,\b{N}}(t)|\psi_s^{(0)}\rangle\\
    &\approx
        \int_{\b{s} \in \mathcal{S}_t^n}\sum\limits_{(\b{j},\b{d})\in\mathcal{J}(\b{N})}
        \left(\prod\limits_{k=1}^nV_{j_k}(s_k)\right)
        \sum\limits_{m=0\atop\text{$m$ even}}^\infty(-1)^\frac{m}{2}
        \int_{\b{\tau} \in \mathcal{S}_t^m}\sum_{\b{\kappa}\in \{1,\cdots,D\}^m}
            \mathcal{L}_b^{\rm same}(\b{\tau}, \b{\kappa})\\
	&\qquad\frac{1}{(2\pi\epsilon)^{3D/2}}\iint_{\b{p},\b{q}\in\mathbb{R}^D}\int_{\b{y}\in\mathbb{R}^D}
		\prod\limits_{l=1}^n Q_{d_l}(s_l,\b{p},\b{q})
		\prod\limits_{\theta=1}^m Q_{\kappa_\theta}(\tau_\theta,\b{p},\b{q})\\
        &\qquad\qquad\qquad\qquad\qquad\qquad\qquad \cdot e^{i\phi(t,\b{x},\b{y},\b{p},\b{q})/\epsilon}
		\psi_s^{(0)}(\b{y})d\b{y}d\b{p}d\b{q}
	d\b{\tau}d\b{s}.
\end{aligned}
\end{equation}

Next, numerical integration is used to evaluate the integral with respect to $\b{p}$, $\b{q}$ in \eqref{eqn:InN(t,x)-FGA}, i.e., consider a set of discrete points and quadrature weights, $\{(\b{p}_k,\b{q}_k,w_k)\}_{k=1}^{N_{\rm qp}}$.
Let
\begin{equation} \label{eqn:psi_k}
	\psi_k(t,\b{x})
	=\frac{1}{(2\pi\epsilon)^{3D/2}}\int_{\b{y}\in\mathbb{R}^D}
		a(t, \b{p}_k, \b{q}_k) e^{i\phi(t,\b{x},\b{y},\b{p}_k,\b{q}_k)/\epsilon}
		\psi_s^{(0)}(\b{y})
		d\b{y}.
\end{equation}
Similarly, we have $\b{P}_k(t)$, $\b{Q}_k(t)$, $S_k(t)$ and $a_k(t)$ subject to \eqref{eqn:FGA-ODEs} with $\b{p}=\b{p}_k$, $\b{q}=\b{q}_k$, which can be evolved individually.

Hence,
\begin{equation}\label{eqn:InN(t,x)-k-FGA}
\begin{aligned}
	I_{n,\b{N}}(t,\b{x})
        &=\sum\limits_{m=0\atop\text{$m$ even}}^\infty I_{n,\b{N}}^{(m)}(t,\b{x})\\
	&\approx \sum\limits_{m=0\atop\text{$m$ even}}^\infty\sum\limits_{k=1}^{N_{\rm qp}}w_k\psi_k(t,\b{x})
        \underbrace{
            \int_{\b{s} \in \mathcal{S}_t^n}\sum\limits_{(\b{j},\b{d})\in\mathcal{J}(\b{N})}
            \left(\prod\limits_{l=1}^n V_{j_l}(s_l)Q_{k,{d_l}}(s_l)\right) d\b{s}
        }_{=:J_{k,\b{N}}(t)}\\
	&\qquad\underbrace{(-1)^\frac{m}{2}
        \int_{\b{\tau} \in \mathcal{S}_t^m}\sum_{\b{\kappa}\in \{1,\cdots,D\}^m}
    	\mathcal{L}_b^{\rm same}(\b{\tau}, \b{\kappa})
    	\prod\limits_{\theta=1}^m Q_{k,\kappa_\theta}(\tau_\theta) d\b{\tau}
    }_{=:J_k^{(m)}(t)}.
\end{aligned}
\end{equation}
\begin{remark}
    We note that it is numerically unaffordable to apply the Inchworm FGA method in \cite{wang2025solving} to two-dimensional cases due to the computational cost of $O(N_{\rm qp}^2)$ for $N_{\rm qp}$ pairs of $(\b{p}_k, \b{q}_k)$.
    To maintain the accuracy of FGA, it is necessary to choose $N_{\rm qp} \geq (L/\sqrt{\epsilon})^4$. \\
    For example, when $\epsilon=1/64$, even with a small domain size of $L=4$, we need $N_{\rm qp} \geq 32^4$. This requires $\sim 1$TB to store the intermediate variables as single-precision complex numbers.
    
	In contrast, our approach only costs $O(N_{\rm qp})$, with some loss in accuracy resulting from the inclusion of fewer terms in the Dyson series compared to \cite{wang2025solving}.
\end{remark}
\begin{remark}
	The space complexity of $I_{n,\b{N}}(t,\b{x})$ is
	\begin{equation*}
		\binom{rD+n-1}{n}
		\approx\frac{(rD)^n}{n!},
	\end{equation*}
	when $rD$ is relatively larger than $n$. Moreover, the total space complexity asymptotically becomes
	\begin{equation*}
		\sum\limits_{n=0}^\infty\frac{(rD)^n}{n!}
		=\exp(rD).
	\end{equation*}
\end{remark}
The FGA ansatz transforms the interaction operator to a scalar, allowing for an exchange of variables in the integral. Based on this, two simplifications are derived in the next subsection, after which the whole process is concluded as an algorithm.

\subsection{An efficient algorithm for solving density of the open quantum system}\label{subsec:alg}
We note that the $n$- and $m$-dimensional integrals in \eqref{eqn:InN(t,x)-k-FGA} limit the number of terms in the Dyson series that can be computed.
In this subsection, we focus on simplifying these integrals to obtain expressions that are independent of the number $n = |\b{N}|$ in the Dyson series.

First, we reduce the $n$-dimensional integral $J_{k,\b{N}}(t)$ to a one-dimensional integral.
\begin{proposition}\label{prop:cross}
	The integral in the first line of \eqref{eqn:InN(t,x)-k-FGA} can be simplified as
	\begin{equation}\label{eqn:JkN(t)}
	\begin{aligned}
		J_{k,\b{N}}(t)
		&:=
            \int_{\b{s} \in \mathcal{S}_t^{n}}\sum\limits_{(\b{j},\b{d})\in\mathcal{J}(\b{N})}
			\left(\prod\limits_{l=1}^{n} V_{j_l}(s_l)Q_{k,{d_l}}(s_l)\right)
            d\b{s}
		=\frac{1}{\b{N}!}\prod\limits_{d=1}^D\prod\limits_{j=1}^{r}\left(I_k^{(j,d)}(t)\right)^{N_j^{(d)}},
	\end{aligned}
	\end{equation}
	where
	\begin{equation}\label{eqn:Ik-d,j}
		I_k^{(j,d)}(t)
		:=\int_0^t
        V_j(s)Q_{k,d}(s)ds.
	\end{equation}
\end{proposition}
\begin{proof}
    By the definition of $\mathcal{J}(\b{N})$ in \eqref{eqn:JnN}, $(\b{j},\b{d})\in\mathcal{J}(\b{N})$ covers all the permutations for a specific $\b{s}$, hence the integral is symmetric with respect to the exchange of variables. We have
	\begin{equation*}
	\begin{aligned}
        J_{k,\b{N}}(t)
        &=\int_{\b{s} \in \mathcal{S}_t^n}
			\frac{1}{n!} \sum_{\sigma \in \mathscr{Q}_n}
                \sum\limits_{(\b{j},\b{d})\in\mathcal{J}(\b{N})}
			\left(\prod\limits_{l=1}^nV_{j_{\sigma(l)}}(s_l)Q_{k,{d_{\sigma(l)}}}(s_l)\right)d\b{s}\\
        &=\int_{\b{s} \in \mathcal{S}_t^n}
			\frac{1}{n!} \sum_{\sigma \in \mathscr{Q}_n}
                \sum\limits_{(\b{j},\b{d})\in\mathcal{J}(\b{N})}
			\left(\prod\limits_{l=1}^nV_{j_l}(s_{\sigma^{-1}(l)})Q_{k,{d_l}}(s_{\sigma^{-1}(l)})\right)d\b{s}\\
		&=\frac{1}{n!}\int_0^t \cdots \int_0^t
                \sum\limits_{(\b{j},\b{d})\in\mathcal{J}(\b{N})}
			\left(\prod\limits_{l=1}^nV_{j_l}(s_l)Q_{k,{d_l}}(s_l)\right)ds_1 \cdots ds_n\\
		&=\frac{1}{n!}\frac{n!}{\b{N}!}
			\prod\limits_{d=1}^D\prod\limits_{j=1}^{r} 
            \left(\int_0^t V_j(s)Q_{k,d}(s)ds\right)^{N_j^{(d)}}
		=\frac{1}{\b{N}!}\prod\limits_{d=1}^D\prod\limits_{j=1}^{r}\left(I_k^{(j,d)}(t)\right)^{N_j^{(d)}}.
	\end{aligned}
	\end{equation*}
\end{proof}
\begin{remark}
	Computation of \eqref{eqn:Ik-d,j} only requires discrete values of $V_j$ when introducing temporal discretization. In fact, using $N_t$ timesteps of equal size gives
	\begin{equation*}
		\tilde{B}=[b_{jk}],
		\qquad
		b_{jk}
		=B(k\Delta t,j\Delta t),
		\qquad
		\Delta t=\frac{t}{N_t},
	\end{equation*}
	which is a Hermitian matrix. Therefore, $V_j(s)$ in the numerical integration is an element of the truncation of the spectral expansion of the discrete two-point correlation function $\tilde{B}$.
\end{remark}
Next, we reduce the $m$-dimensional integral $J_k^{(m)}(t)$ to a two-dimensional integral, using a similar method to \cite{liu2025error}.
\begin{proposition}\label{prop:same}
	In \eqref{eqn:InN(t,x)-k-FGA}, the integral on the single interval $[0,t]$ can be rewritten as
	\begin{equation}\label{eqn:Jk(m)(t)}
	\begin{aligned}
		J_k^{(m)}(t)
		&:=(-1)^\frac{m}{2}
            \int_{\b{\tau} \in \mathcal{S}_t^m}\sum_{\b{\kappa}\in \{1,\ldots,D\}^m}
            \mathcal{L}_b^{\rm same}(\b{\tau}, \b{\kappa})\prod\limits_{\theta=1}^mQ_{k,\kappa_\theta}(\tau_\theta)d\b{\tau}\\
		&=\frac{(-1)^\frac{m}{2}}{(m/2)!}\left(\int_{0\le\tau_1\le \tau_2\le t}B(\tau_2,\tau_1)\left(\sum\limits_{d=1}^DQ_{k,d}(\tau_1)Q_{k,d}(\tau_2)\right)d\tau_1d\tau_2\right)^\frac{m}{2}
		=\frac{\left(J_k^{(2)}(t)\right)^\frac{m}{2}}{(m/2)!}.
	\end{aligned}
	\end{equation}
\end{proposition}
\begin{proof}
	Noting the contribution of the delta function in \eqref{eqn:Lb nd}, one can derive
	\begin{equation*}
		J_k^{(m)}(t)
		=(-1)^\frac{m}{2}
        \int_{\b{\tau} \in \mathcal{S}_t^m}
			\sum\limits_{\b{q}\in\mathscr{L}_s(\b{\tau})}
				\prod\limits_{(t_1,t_2)\in\b{q}}\Big(\b{Q}_k(t_1)\cdot\b{Q}_k(t_2)\Big)B(t_2,t_1)
				d\b{\tau}.
	\end{equation*}
	For even $m\ge2$, let $\mathcal{P}_k(t_1,t_2)=(\b{Q}_k(t_1)\cdot\b{Q}_k(t_2))B(t_2,t_1)$ and $\b{\tau}_m=(\tau_1,\tau_2,\ldots,\tau_m)$, then 
	\begin{equation*}
	\begin{aligned}
		&(-1)^\frac{m}{2}J_k^{(m)}(t)\cdot(-1)^\frac{2}{2}J_k^{(2)}(t)\\
        &\int_{\b{\tau}_m \in \mathcal{S}_t^m}
        \int_{0\le\tau_{m+1}\le\tau_{m+2}\le t}		\left(\sum\limits_{\b{q}\in\mathscr{L}_s(\b{\tau}_m)}\prod\limits_{(t_1,t_2)\in\b{q}}\mathcal{P}_k(t_1,t_2)\right)
			\mathcal{P}_k(\tau_{m+1},\tau_{m+2})
			d\b{\tau}_{m+2}\\
		&=\frac{\binom{m+2}{2}(m-1)!!}{(m+1)!!}
			\cdot(-1)^\frac{m+2}{2}J_k^{(m+2)}(t).
	\end{aligned}
	\end{equation*}
	Specifically, adding one additional pair to the diagram consisting of $m/2$ pairs produces $\binom{m+2}{2}$ new diagrams. Conversely, each diagram of $m/2+1$ pairs generates $m/2+1$ diagrams when removing one pair, meaning that the factor in the second equality above can be computed by the quotient between diagram numbers.
	Therefore
	\begin{equation*}
		J_k^{(m)}(t)
		=\frac{J_k^{(m-2)}(t)J_k^{(2)}(t)}{m/2}
		=\frac{J_k^{(m-4)}\left(J_k^{(2)}(t)\right)^2}{(m/2)(m/2-1)}
		=\cdots
		=\frac{\left(J_k^{(2)}(t)\right)^\frac{m}{2}}{(m/2)!}.
	\end{equation*}
\end{proof}

Using \eqref{eqn:JkN(t)} and \eqref{eqn:Jk(m)(t)}, instead of evaluating high-dimensional integrals, the computational complexity of $J_{k,\b{N}}(t)$ is reduced to $O(rDN_t)$, while that of $J_k^{(m)}(t)$ is now $O(DN_t^2)$.

Finally, by truncating the Dyson series, i.e. considering $n\le\bar{N}$, we obtain
\begin{align}
	\label{eqn:rho(t,x)}
	\rho_s(t,\b{x})
    &=\sum\limits_{n=0}^{\bar{N}}\sum\limits_{|\b{N}|=n} \b{\lambda}^{\b{N}} \b{N}!
		\sum\limits_{m_1+m_2\le\bar{N}-2n\atop\text{$m_1$, $m_2$ even}}
			\left(I_{n,\b{N}}^{(m_1)}(t,\b{x})\right)
			\left(I_{n,\b{N}}^{(m_2)}(t,\b{x})\right)^*,\\
	\label{eqn:InN(m)(t,x)}
	I_{n,\b{N}}^{(m)}(t,\b{x})
	&=\sum\limits_{k=1}^{N_{\rm qp}}w_k\psi_k(t,\b{x})J_{k,\b{N}}(t)J_k^{(m)}(t).
\end{align}

We conclude this section with the following algorithm to compute the density $\rho_s(T,\b{x})$ in a discretized domain $\mathcal{X}$, at a single time point $T$.
\renewcommand{\algorithmicrequire}{\textbf{Input:}}
\renewcommand{\algorithmicensure}{\textbf{Output:}}
\begin{algorithm}[H]
	\caption{Computation of density for Caldeira-Leggett model.}
	\begin{algorithmic}
		\REQUIRE Time $T = N_t\Delta t$, \\
        Low rank approximation $\{\lambda_j,V_j\}_{j=1}^{r}$ for discrete values of two-point correlation function, \\
		Truncation order $\bar{N}$ for Dyson series, \\
        $N_{\rm qp}$ pairs of quadrature points $(\b{p}_k,\b{q}_k)$ and weights $w_k$ for double integral over $\b{p}, \b{q}$, \\
        $N_y$ quadrature points and weights for integral over $\b{y}$.
		\ENSURE Density $\rho_s(T,\b{x})$ for $N_x$ discrete points $\{\b{x}\}$ in $\mathcal{X}$.
        \STATE
		\STATE Initialize $I_{n,\b{N}}^{(m)}(T,\b{x})$ for $n=0,1,\dots,\bar{N}$ and $m=0,2,\dots,2\bar{N}-2n$.
		\FOR{$k$ from 1 to $N_{\rm qp}$}
        \STATE Initialize FGA variables with $\b{p}_k, \b{q}_k$.
		\STATE Perform FGA evolution and record $\{\b{Q}_k(j\Delta t)\}_{j=0}^{N_{\rm t}}$.       
        \STATE Compute $J_{k,\b{N}}(T)$ via \eqref{eqn:JkN(t)} for $n=0,1,\dots,\bar{N}$.
        \STATE Compute $J_k^{(m)}(T)$ via \eqref{eqn:Jk(m)(t)} for $m=0,2,\dots,2\bar{N}$.
		\FOR{$\b{x}$ in $\mathcal{X}$}
        \STATE Compute $\psi_k(T,\b{x})$ via discretization of \eqref{eqn:psi_k}.
		\STATE Add contribution to $I_{n,\b{N}}^{(m)}(T,\b{x})$ via \eqref{eqn:InN(m)(t,x)} for $n=0,1,\dots,\bar{N}$ and $m=0,2,\dots,2\bar{N}-2n$.
		\ENDFOR
		\ENDFOR
		\STATE Compute $\rho_s(T,\b{x})$ via \eqref{eqn:rho(t,x)} for all $\b{x}$ in $\mathcal{X}$.
	\end{algorithmic}
\end{algorithm}
\begin{remark}
	The low-rank approximation of the two-point correlation function can be generated by a low-rank approximation of the matrix corresponding to the discrete two-point function. It is actually a truncated spectral expansion of that matrix.
\end{remark}
We evaluate the overall space and time complexities of this algorithm as follows.
\begin{itemize}
    \item All values of $I_{n,\b{N}}^{(m)}(T,\b{x})$ are stored, with total space complexity
    \begin{equation}
        O((\bar{N}+rD)^{\bar{N}}N_x)
    \end{equation}
    that dominates all other intermediate variables. $\bar{N}$ is the truncation order of the Dyson series, $r$ is the rank of the low-rank approximation, and $N_x$ is the number of points in $\mathcal{X}$ that are considered.
    \item For each value of $k$:
    \begin{enumerate}
        \item Using a linear multi-step method for FGA evolution, the time complexity is $O(N_t + N_xN_y)$.
        \item Computing $J_{k,\b{N}}(T)$ and $J_k^{(m)}(T)$ costs $O(rDN_t + DN_t^2) = O(DN_t^2)$ when $r < N_t$.
        \item Adding a contribution to every $I_{n,\b{N}}^{(m)}(T,\b{x})$ costs $O((\bar{N}+rD)^{\bar{N}}N_x)$.
    \end{enumerate}
    Iterating over $N_{\rm qp}$ values of $k$ gives a total time complexity of
    \begin{equation}
        O(N_{\rm qp}(N_xN_y + DN_t^2 +(\bar{N}+rD)^{\bar{N}}N_x)).
    \end{equation}
    The final loop to compute all values of $\rho_s(T,\b{x})$ has time complexity \\
    $O(\bar{N}(\bar{N}+rD)^{\bar{N}}N_x)$, which is dominated by the previous parts as usually $\bar{N} \ll N_{\rm qp}$.
    
\end{itemize}
\section{Numerical results}\label{sec:num result}
In this section, several numerical experiments are performed to validate the proposed method. Firstly, a one-dimensional double well is demonstrated to validate our algorithm by recovering numerical results in the existing literature. Next, a one-dimensional harmonic oscillator is presented to display the effectiveness of the low-rank approximation, followed by a two-dimensional version in which the anticipated density can be observed. Finally, a two-dimensional double slit is simulated.

The two-point correlation function can be formulated as
\begin{equation}
	B(\tau_1,\tau_2)
	=\tilde{B}(\Delta\tau)
	=\frac{1}{2}\sum\limits_{l=1}^L\frac{c_l^2}{\epsilon\omega_l}
		\left(\coth\left(\frac{\beta\epsilon\omega_l}{2}\right)\cos(\omega_l\Delta\tau)-i\sin(\omega_l\Delta\tau)\right),
\end{equation}
where $\Delta\tau=\tau_1-\tau_2$.
For all our simulations, we use the Ohmic spectral density \cite{chakravarty1984dynamics,makri1999linear}, where the frequencies $\omega_l$ and the coupling intensities $c_l$ are given by
\begin{align}
	\omega_l
	&=-\omega_c\log\left(1-\frac{l}{L}(1-\exp(-\omega_{\rm max}/\omega_c))\right),\\
	c_l
	&=\epsilon\omega_l\sqrt{\frac{\xi\omega_c}{L}(1-\exp(-\omega_{\rm max}/\omega_c))},
\end{align}
for $l=1,2,\ldots,L$. Additionally, we choose the parameters to be
\begin{equation}
	L=400,
	\qquad
	\omega_{\rm max}=10,
	\qquad
	\omega_c=2.5,
	\qquad
	\beta=5.
\end{equation}
In this paper, a low-rank approximation is numerically attained for the discrete two-point correlation function. An analog focusing on the continuous form may be obtained via the discrete Lehmann representation \cite{kaye2022discrete,kiese2025discrete}.

For all of our experiments, the effective potential energy $V$ will also include an extra term for homogeneous dissipation in space,
\begin{equation}\label{eqn:tilde V}
    \tilde{V}(\b{x}) = \frac{1}{2} \omega_b^2 |\b{x}|^2,
    \qquad \text{where } \omega_b^2 = \sum_{l=1}^L \frac{c_l^2}{\omega_l^2}.
\end{equation}

\subsection{Double well}
We consider the double well potential
\begin{equation}
	V(x)
	=-x^2+2x^4.
\end{equation}
The initial wave function is given by
\begin{equation}
	\psi_1(x)
	=C_{\rm nor }\left(\exp\left(-\frac{(x-1/2)^2}{4\epsilon}\right)+\frac{4}{5}\exp\left(-\frac{(x+1/2)^2}{4\epsilon}\right)\right),
\end{equation}
where $C_{\rm nor}=5(41+40^{-\frac{1}{8\epsilon}})^{-1/2}(2\pi\epsilon)^{-1/4}$.
The same configurations as those in \cite{wang2025solving} are adopted. Specifically, we take $\epsilon=1/64$ and the truncated domains for $\b{p}, \b{q}$ to be $[-2,2]$ with increments $\Delta p=\Delta q=1/32$ in each dimension.

Numerical experiments with $\xi=0$ (``without bath''),$1.6,3.2,6.4$ were carried out to show the influence of bath when the coupling effect is enhanced. We choose $r=20$ for the low-rank approximation, with Frobenius norms $\|B_{\rm LR}-\tilde{B}\|_F=1.0906\times10^{-10},2.1827\times10^{-10},4.3624\times10^{-10}$ between the approximation and the discrete two-point correlation function for $\xi=1.6,3.2,6.4$, respectively.

The densities for $\bar{N}=5$ when $t=1,2,3$ are shown in Fig. \ref{fig:den-dif xi-doubleWell}. Using the numerical results with $\bar{N}=5$ as a reference, the $L^2$ density errors are listed in Tab. \ref{tab:den err-doubleWell}. 
Lastly, with $\xi=6.4$, the densities for $\bar{N}=1,2,3,4,5$ and the density differences $\Delta\rho^{(D)}=\rho^{\bar{N}=D}-\rho^{\bar{N}=D-1}$ for $D=2,3,4,5$ when $t = 2,3$ are presented in Figs. \ref{fig:den-dif n-xi6.4-t=2-doubleWell} and \ref{fig:den-dif n-xi6.4-t=3-doubleWell}, respectively.

\begin{figure}[H]
    \centering
    \includegraphics[width=.32\linewidth]{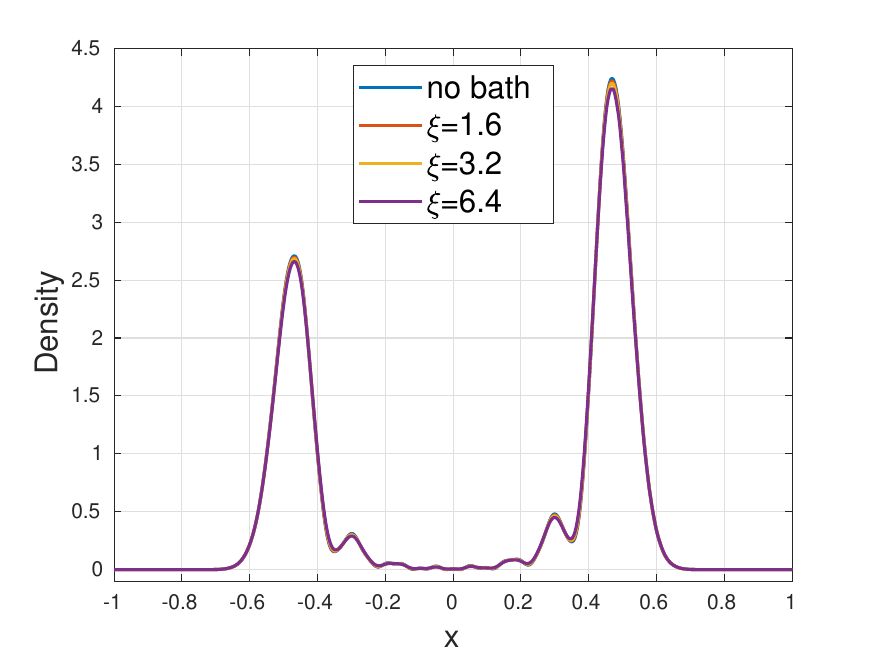}
    \includegraphics[width=.32\linewidth]{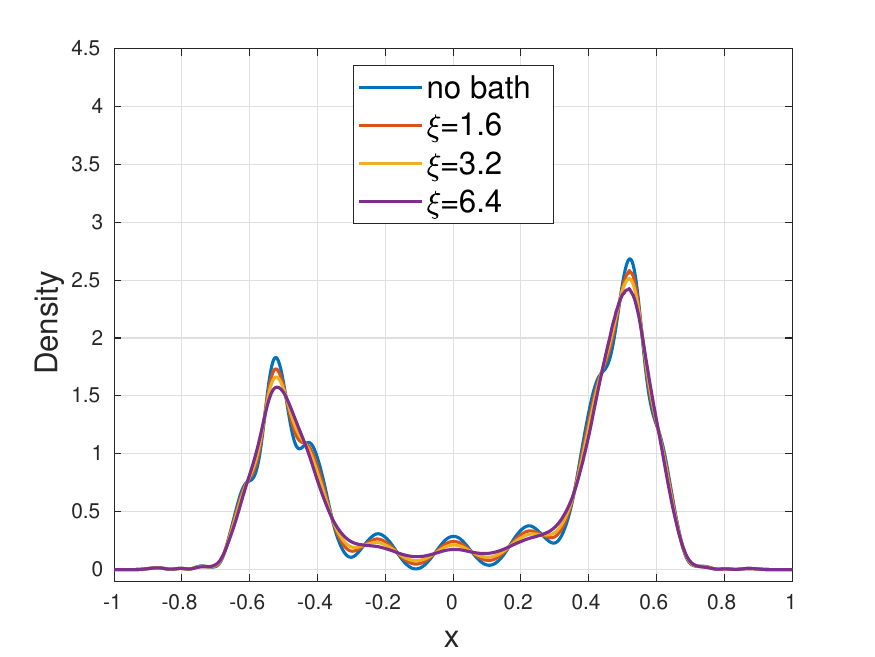}
    \includegraphics[width=.32\linewidth]{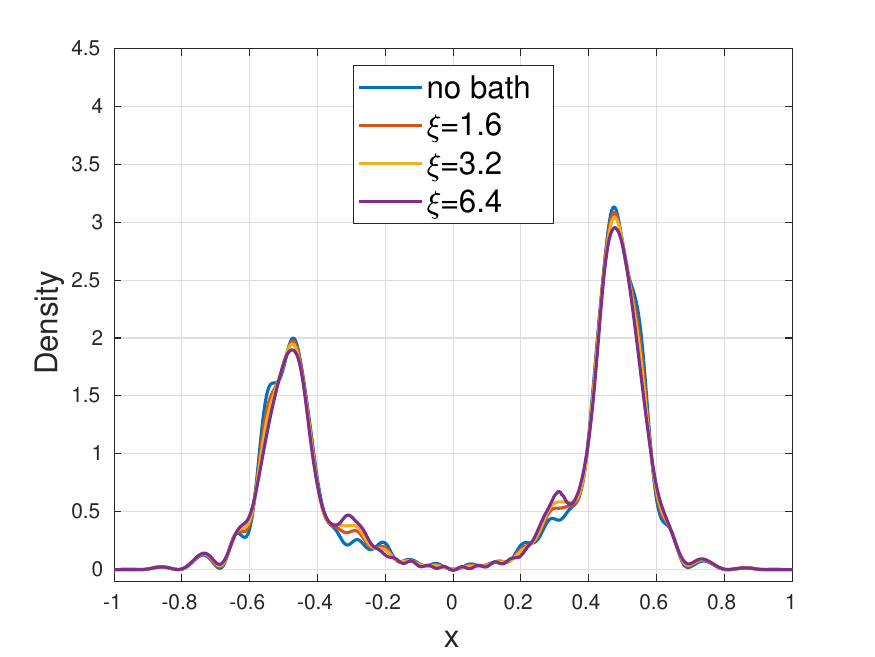}
    \caption{Densities for $\bar{N}=5$ when $t=1$ (left), $t=2$ (middle), $t=3$ (right).}
    \label{fig:den-dif xi-doubleWell}
\end{figure}
In Fig. \ref{fig:den-dif xi-doubleWell}, quantum decoherence can be successfully observed. This also recovers the results in \cite{wang2025solving}. Furthermore, a smoother behavior can be observed as $\xi$ increases, which is consistent with the fact that $\xi$ depicts the strength of the coupling effect between the system and bath.
\begin{table}[H]
    \centering
    \caption{Density errors $\|\rho-\rho^{\bar{N}=5}\|_{L^2}$ for different $\xi$ and $t$.}
    \label{tab:den err-doubleWell}
	\begin{tabular}{c|c|cccc}
		\hline
        \rule{0pt}{1em}
		\multirow{2}{*}{$\xi$}	&\multirow{2}{*}{$t$}	&\multicolumn{4}{c}{$\bar{N}$}\\\cline{3-6}
		&		&1	&2	&3	&4\\
		\hline
		\multirow{3}{*}{1.6}
		&1	&1.2417e-04		&1.1462e-05		&8.3204e-07		&5.1514e-08\\
		&2	&5.2847e-03		&7.1362e-04		&7.2292e-05		&6.4204e-06\\
		&3	&8.4033e-03		&1.7152e-03		&2.6092e-04		&3.7885e-05\\
		\hline
		\multirow{3}{*}{3.2}
		&1	&4.6202e-04		&8.6527e-05		&1.2561e-05		&1.6550e-06\\
		&2	&1.8835e-02		&5.2462e-03		&1.0558e-03		&2.0560e-04\\
		&3	&2.8667e-02		&1.2361e-02		&3.6203e-03		&1.2214e-03\\
		\hline
		\multirow{3}{*}{6.4}
		&1	&1.6015e-03		&6.2208e-04		&1.7543e-04		&5.3175e-05\\
		&2	&5.9581e-02		&3.6806e-02		&1.3621e-02		&6.5711e-03\\
		&3	&7.9739e-01		&9.0164e-01		&4.0157e-01		&3.9711e-01\\
		\hline
	\end{tabular}
\end{table}
Convergence to the reference can be observed in Tab. \ref{tab:den err-doubleWell}. In particular, a larger $\xi$ leads to larger errors, resulting from the corresponding contribution in the two-point correlation function that influences the convergence of the Dyson series. To further investigate the error distribution when $\xi=6.4$, the densities and density differences between adjacent $\bar{N}$ are presented in the following two figures.
\begin{figure}[H]
    \centering
    \includegraphics[width=.49\linewidth]{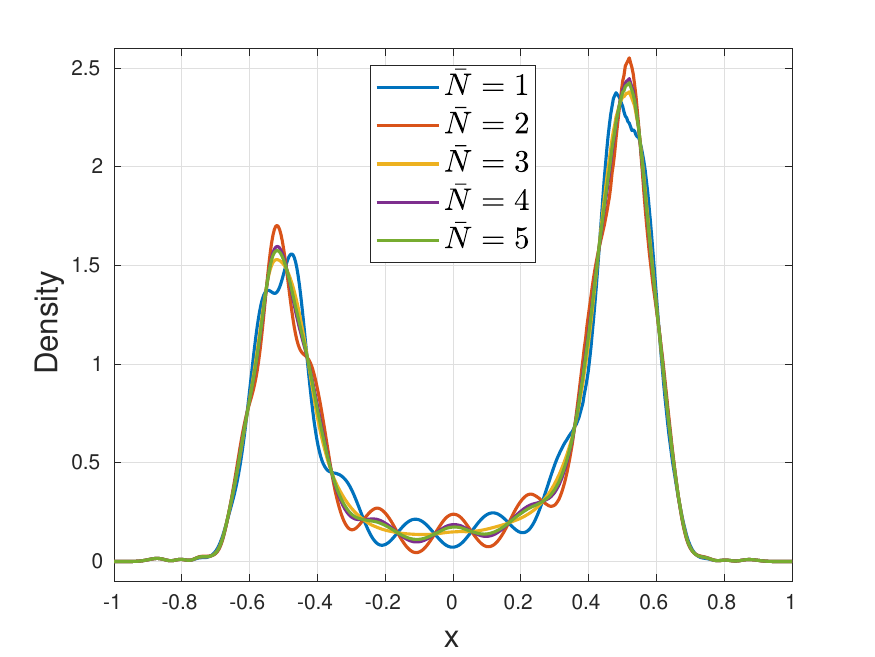}
    \includegraphics[width=.49\linewidth]{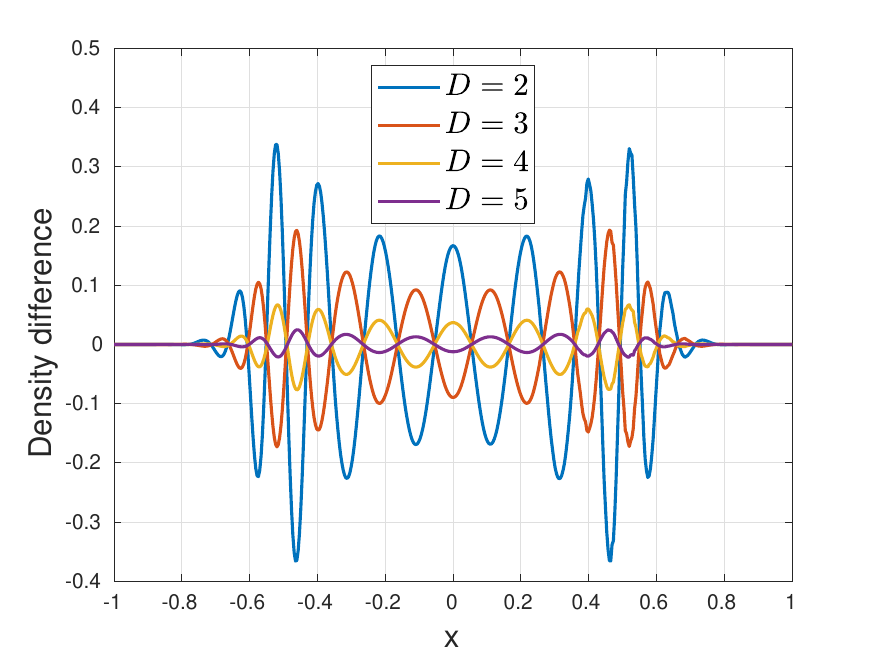}
    \caption{Density (left)/ Density difference (right) when $\xi=6.4$ and $t=2$.}
    \label{fig:den-dif n-xi6.4-t=2-doubleWell}
\end{figure}
\begin{figure}[H]
    \centering
    \includegraphics[width=.49\linewidth]{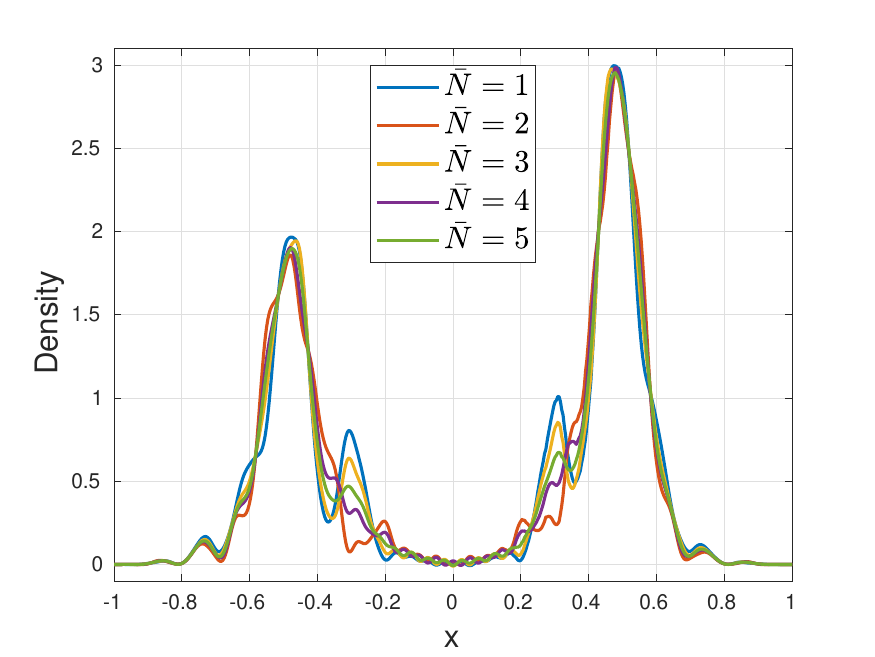}
    \includegraphics[width=.49\linewidth]{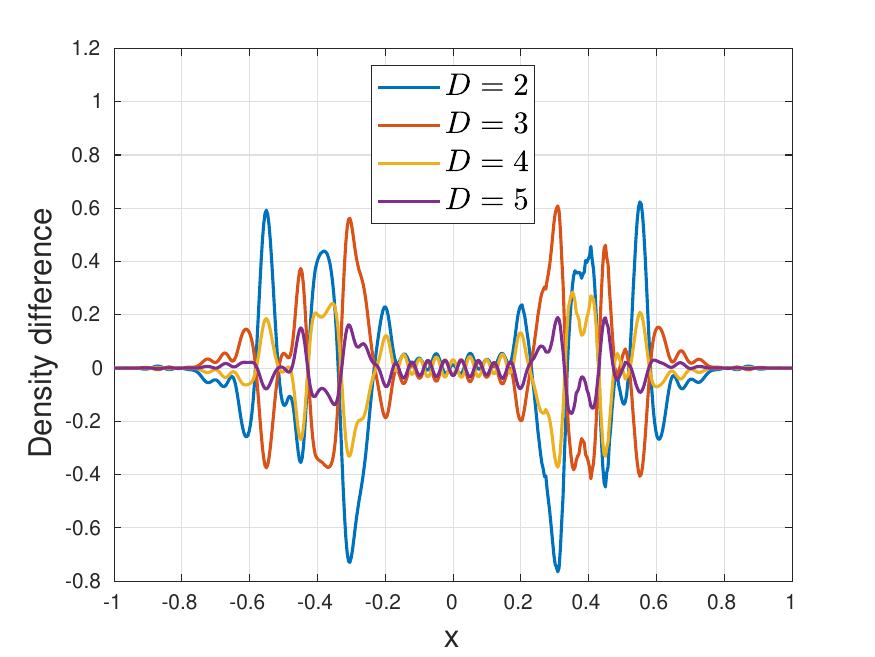}
    \caption{Density (left)/ Density difference (right) when $\xi=6.4$ and $t=3$.}
    \label{fig:den-dif n-xi6.4-t=3-doubleWell}
\end{figure}
It can be observed in Figs. \ref{fig:den-dif n-xi6.4-t=2-doubleWell} and \ref{fig:den-dif n-xi6.4-t=3-doubleWell} that the density difference converges to 0 when $\bar{N}$ increases. The results for $t=3$ demonstrate a larger error due to the longer temporal evolution accumulating more numerical errors, which is severe for a Markovian process.

\subsection{Harmonic oscillator}
In this subsection, we consider the potential of the harmonic oscillator
\begin{equation}
	V(\b{x})
	=\frac{1}{2}|\b{x}|^2.
\end{equation}
We will first confirm the validity of our results for weak coupling cases by comparing with the solution to the Lindblad equation, and then carry out one- and two-dimensional simulations for stronger system-bath coupling.

\subsubsection{Comparison to Lindblad solution}
To validate the proposed method, we first compare the numerical solution from our method to the solution of the Lindblad equation
\begin{equation}
    \frac{d\rho_s}{dt}
	=-\frac{i}{\epsilon}[H_s+H_{LS},\rho_s]+\sum_{\Delta=\pm\epsilon\omega}\gamma(\Delta)
	\left[A(\Delta)\rho_sA^\dagger(\Delta)-\frac{1}{2}\left\{A^\dagger(\Delta)A(\Delta),\rho_s\right\}\right].
\end{equation}
We represent this solution as
\begin{equation*}
    \rho_s(x)
    =\sum_{n,m=0}^{\infty} \rho_{n,m}(t)|\psi_n\rangle\langle\psi_m|,
\end{equation*}
where $|\psi_n\rangle$ is the $n$-th eigenstate of the system with harmonic oscillator potential $V(x)=\omega^2x^2/2$, with $\omega^2=1+\omega_b^2$. We also define $a=\gamma(\epsilon\omega)\cdot\epsilon/(2\omega)$ and $b=\gamma(-\epsilon\omega)\cdot\epsilon/(2\omega)$. We consider two initial states, the ground state and the 1st excited state, where the Lindblad equation solutions are
\begin{equation*}
\begin{gathered}
    \rho_{n,m}(t)=\delta_{n,m}\cdot\Big(1-\alpha\Big)\alpha^n,
    \quad
    \alpha=\frac{a+ae^{(a-b)t}}{b+ae^{(a-b)t}},
    \quad\text{when}\quad
    \rho_{n,m}(0)=\delta_{n,m}\delta_{n,0},\\
    \rho_{n,m}(t)=\delta_{n,m}\cdot\frac{b-a}{B^2}\Big(E(n+1)q^n+Cnq^{n-1}\Big),
    \quad\text{when}\quad
    \rho_{n,m}(0)=\delta_{n,m}\delta_{n,1},
\end{gathered}    
\end{equation*}
with $B=b-a^{(a-b)t}$, $C=-a+b^{(a-b)t}$, $E=b-b^{(a-b)t}$, $q=(a-a^{(a-b)t})/B$.\\
We compare the densities from our method against the Lindblad equation solutions for $\xi=0.002$, 0.005, 0.01, 0.02, 0.05, 0.1 in the following figures.
\begin{figure}[H]
	\centering
	\includegraphics[width=.75\linewidth]{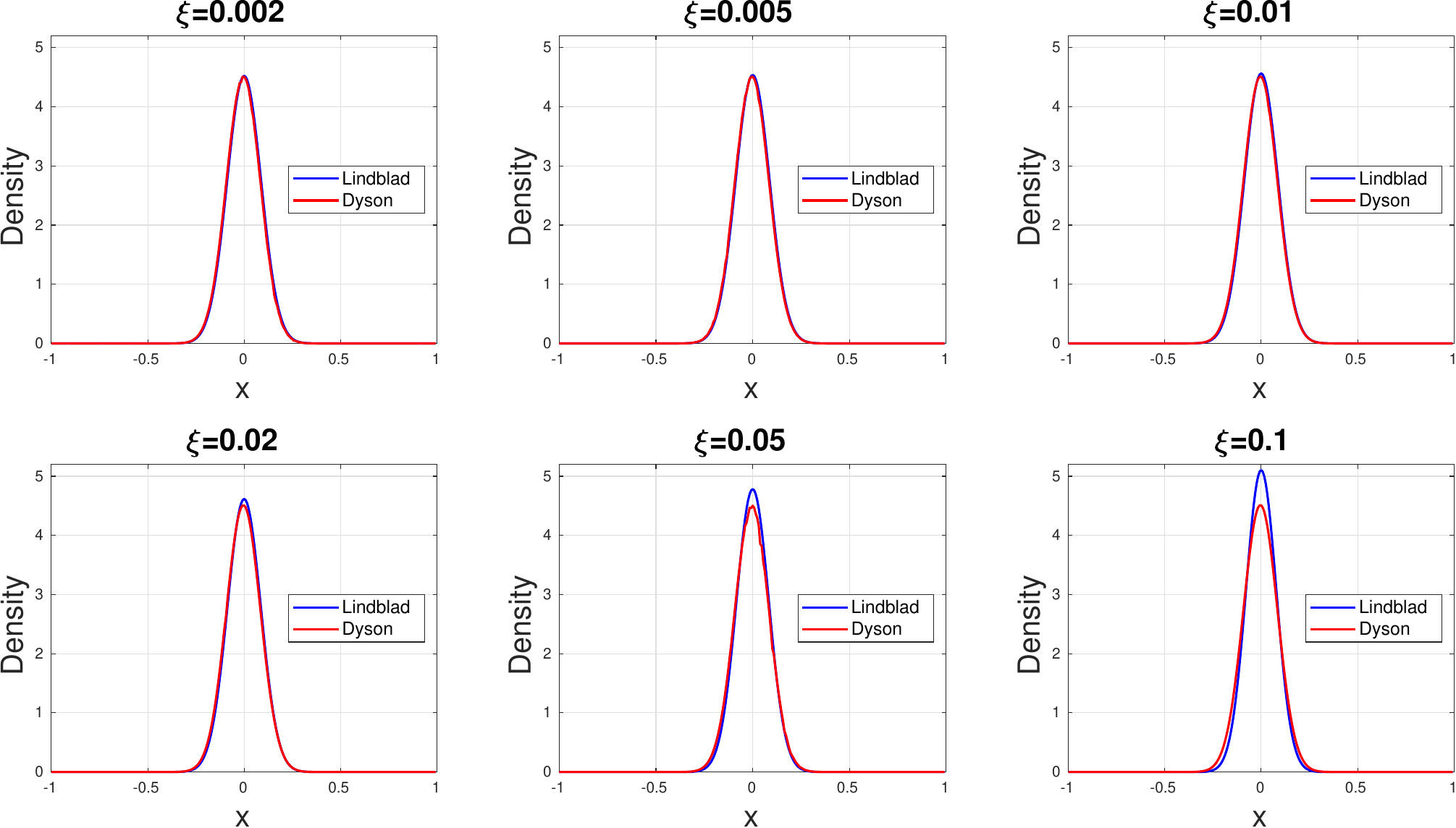}
	\caption{Density comparison when initial state is the ground state.}
	\label{fig:psi0}
\end{figure}
\begin{figure}[H]
	\centering
	\includegraphics[width=.75\linewidth]{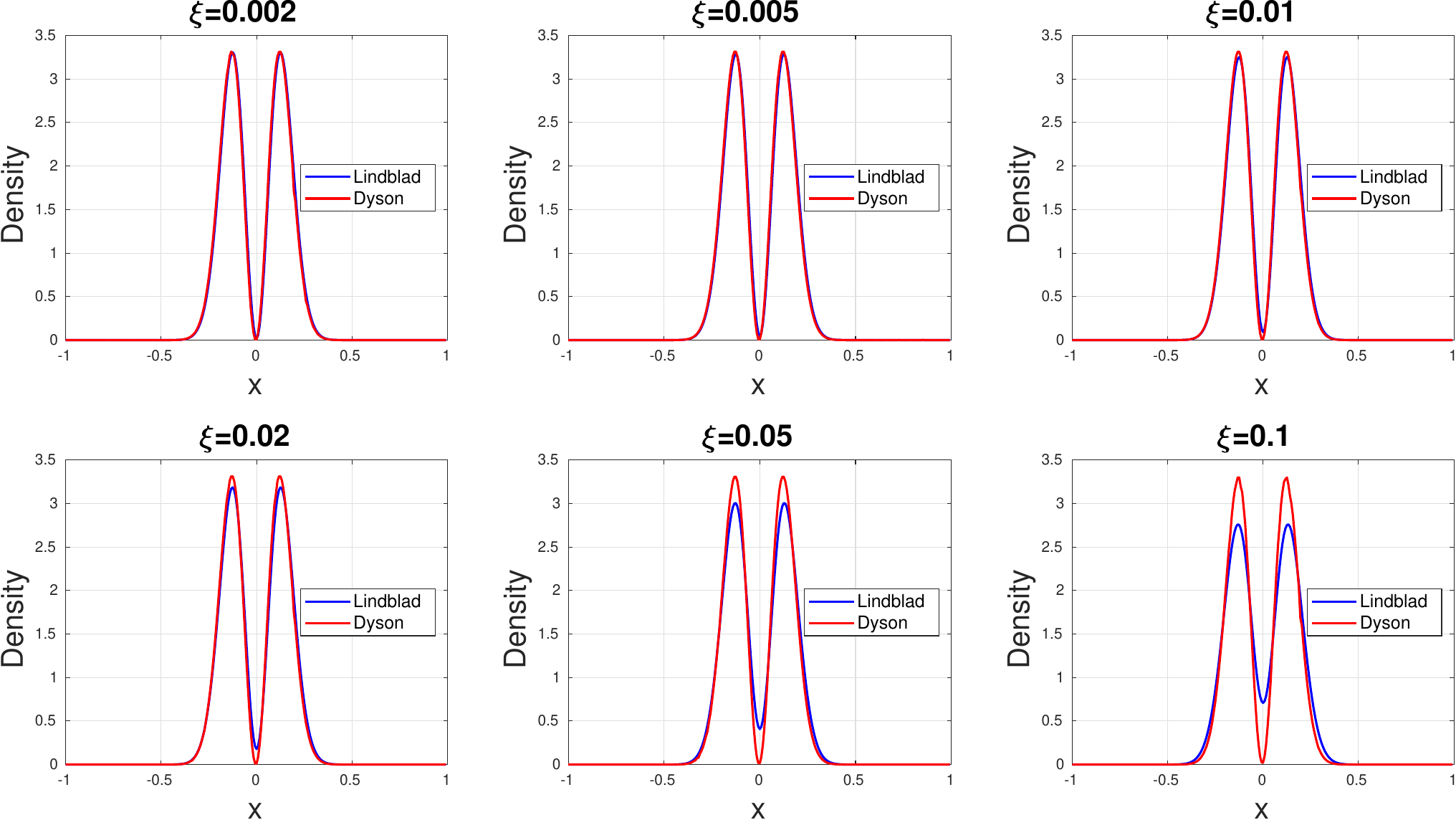}
	\caption{Density comparison when initial state is the 1st excited state.}
	\label{fig:psi1}
\end{figure}
It can be observed in Figs. \ref{fig:psi0} and \ref{fig:psi1} that when $\xi$ is sufficiently small, e.g., $\xi\le 0.01$, the numerical solutions obtained from our method agree well with those produced by the Lindblad equation, with only minor discrepancies. However, as the coupling strength increases, the differences between the two solutions become progressively larger.

For the following experiments, we take $\xi=1.6$.
In addition to the initial value $\psi_1$ defined in the last subsection, we also consider a commonly used initial value
\begin{equation}
	\psi_2(x)
	=\frac{1}{(\pi\epsilon)^{1/4}}\exp\left(-\frac{x^2}{2\epsilon}\right).
\end{equation}
Firstly, a discussion on the validity of low-rank approximation is given in Tab. \ref{tab:lra}. In particular, we test $r=5,10,20,30,40$ for both initial values $\psi_1$ and $\psi_2$. Using the numerical results with $r=40$ as a reference, the integral of density $\int_{-\infty}^\infty \rho(x)dx$ and the $L^2$ difference in density $\|\rho-\rho^{r=40}\|_{L^2}$ is demonstrated in Tab. \ref{tab:lra}. Lastly, the density difference $\rho^{r=5}-\rho^{r=40}$ and the relative density difference $(\rho^{r=5}-\rho^{r=40})/\rho^{r=40}$ for both initial values are presented in Fig. \ref{fig:den err-NLR}.

\subsubsection{Discussion on $r$}
\begin{table}[H]
	\centering
	\caption{Integral of density (top)/ $L^2$ difference in density $\|\rho-\rho^{r=40}\|_{L^2}$ (bottom).}
	\label{tab:lra}
	\begin{tabular}{c|ccccc}
		\hline
        \rule{0pt}{1em}
		\multirow{3}{*}{$r$}	&\multicolumn{5}{c}{$\bar{N}$}
		\\\cline{2-6}
			&1	&2	&3	&4	&5
		\\\cline{2-6}
			&\multicolumn{5}{c}{$\psi_1$}
		\\\hline
		\multirow{2}{*}{5}
		&0.9737    &0.9743    &0.9743    &0.9743    &0.9743\\
		&1.06e-03   &9.41e-04	&9.50e-04	&9.49e-04	&9.49e-04
		\\\hline
		\multirow{2}{*}{10}	
		&0.9736    &0.9742    &0.9742    &0.9742    &0.9742\\
		&3.35e-04    &2.91e-04    &2.97e-04    &2.95e-04    &2.96e-04
		\\\hline
		\multirow{2}{*}{20}	
		&0.9736    &0.9742    &0.9742    &0.9742    &0.9742\\
		&2.88e-09    &2.50e-09    &2.54e-09    &2.53e-09    &2.53e-09
		\\\hline
		\multirow{2}{*}{30}	
		&0.9736    &0.9742    &0.9742    &0.9742    &0.9742\\
		&8.43e-16    &7.99e-16    &8.28e-16    &8.39e-16    &8.46e-16
		\\\hline
		40	
		&0.9736    &0.9742    &0.9742    &0.9742    &0.9742
		\\\hline
        \rule{0pt}{1em}
		\multirow{3}{*}{$r$}	&\multicolumn{5}{c}{$\bar{N}$}
		\\\cline{2-6}
			&1	&2	&3	&4	&5
		\\\cline{2-6}
			&\multicolumn{5}{c}{$\psi_2$}
		\\\hline
		\multirow{2}{*}{5}
			&0.9737    &0.9744    &0.9743    &0.9744    &0.9744\\
			&2.68e-04    &2.44e-04    &2.45e-04    &2.45e-04  &2.45e-04
		\\\hline
		\multirow{2}{*}{10}
			&0.9736    &0.9743    &0.9742    &0.9742    &0.9742\\
			&1.95e-05    &1.67e-05    &1.67e-05    &1.69e-05    &1.69e-05
		\\\hline
		\multirow{2}{*}{20}	
			&0.9736    &0.9743    &0.9742    &0.9742    &0.9742\\
			&2.14e-10    &1.87e-10    &1.89e-10    &1.89e-10    &1.89e-10
		\\\hline
		\multirow{2}{*}{30}	
			&0.9736    &0.9743    &0.9742    &0.9742    &0.9742\\
			&9.23e-16    &9.25e-16    &9.20e-16    &9.23e-16    &9.20e-16
		\\\hline
		40	
			&0.9736    &0.9743    &0.9742    &0.9742    &0.9742
		\\\hline
	\end{tabular}
\end{table}
Tab. \ref{tab:lra} shows that the density integral is well preserved for different configurations, and the density converges to the reference as $r$ increases. Furthermore, all the choices of $r$ result in density errors of less than $1.1\times10^{-3}$.
\begin{figure}[H]
    \centering
    \includegraphics[width=.45\linewidth]{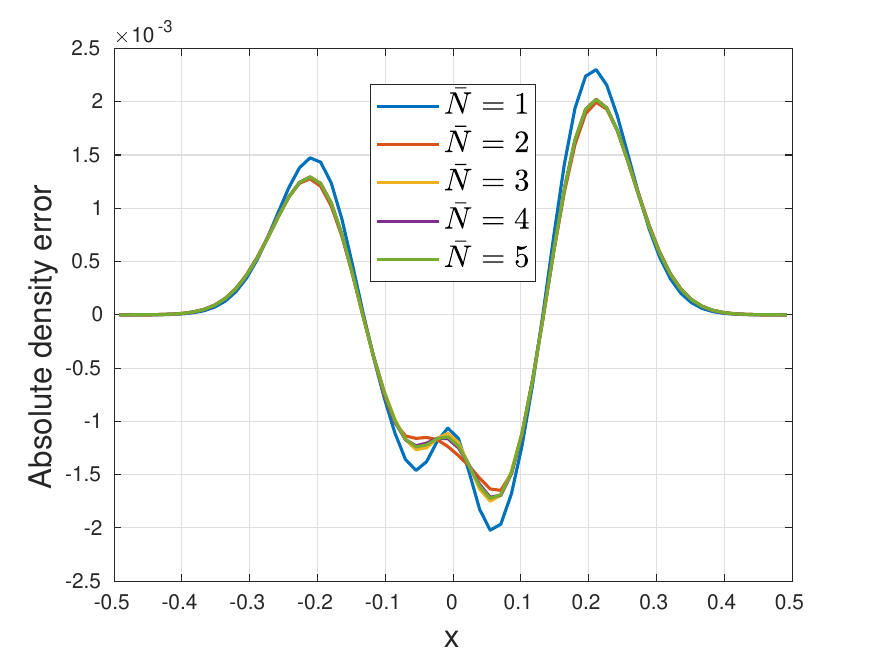}
    \includegraphics[width=.45\linewidth]{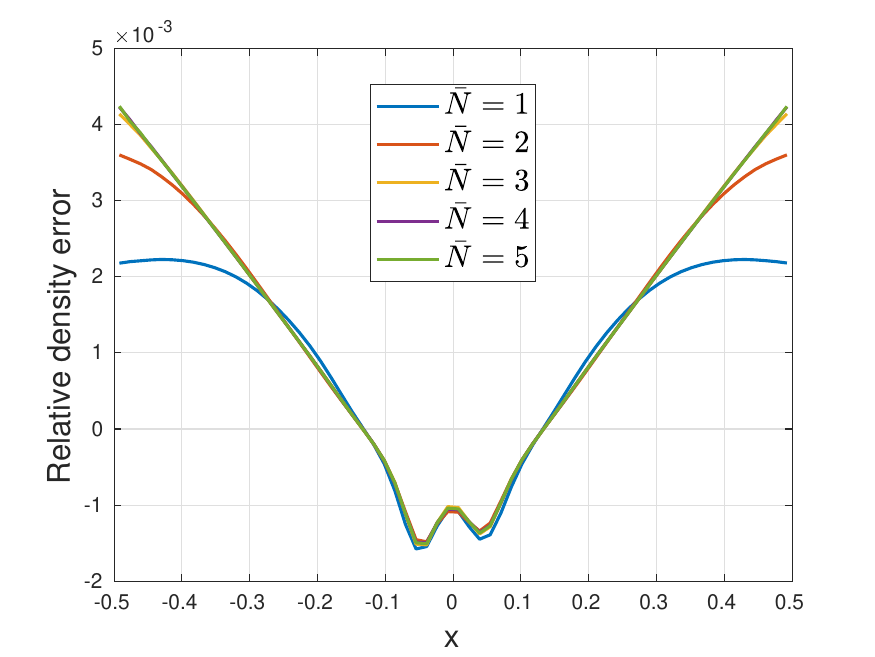}
    \includegraphics[width=.45\linewidth]{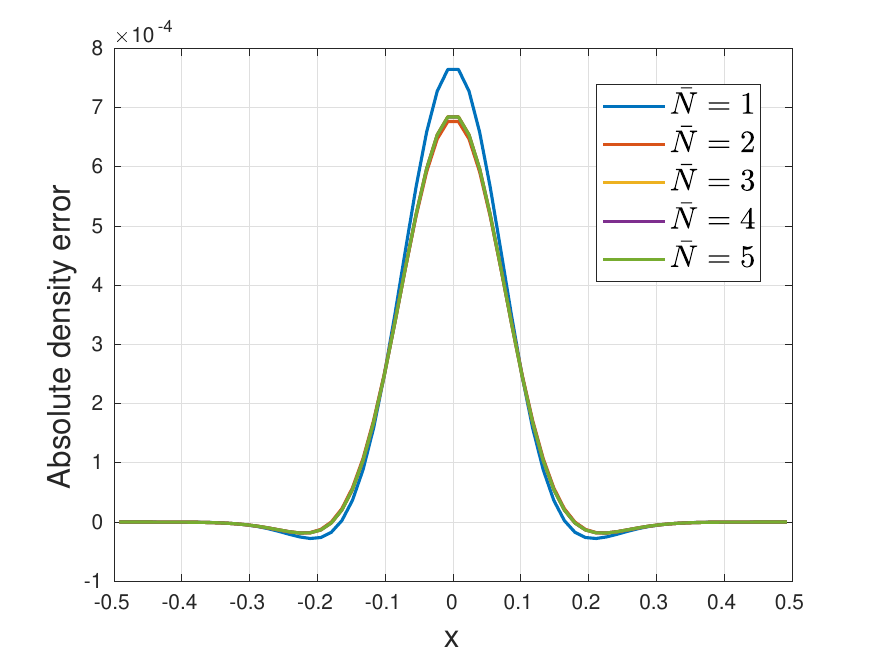}
    \includegraphics[width=.45\linewidth]{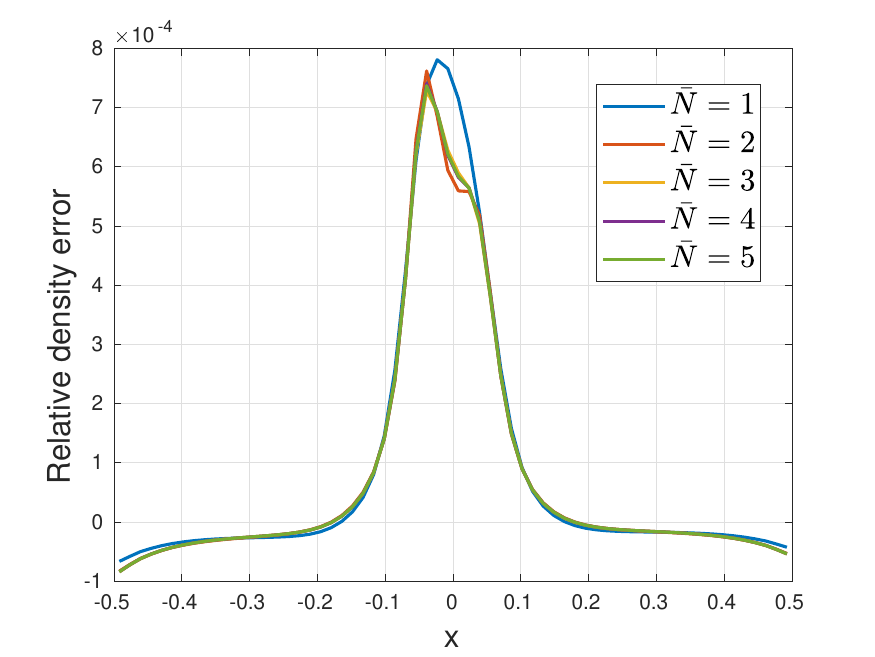}
    \caption{Absolute (left)/ Relative (right) density difference for initial value $\psi_1$ (top)/ $\psi_2$ (bottom).}
    \label{fig:den err-NLR}
\end{figure}
It can be observed in Fig. \ref{fig:den err-NLR} that the local absolute error is less than $5\times10^{-4}$ ($1\times10^{-2}$) for $\bar{N}=1,2,3,4,5$ with initial value $\psi_2$ ($\psi_1$), while the relative error is always of the order $10^{-3}$ in both initial values.

$r=5$ gives a difference in Frobenius norm of $84.3088$, compared to $r=40$ with a difference in norm of $5.4535\times10^{-13}$, while the computational cost for $r=5$ is about $(5/40)^5=2^{-15}\approx3\times10^{-5}$ times the cost for $r=40$. This justifies the choice of $r=5$ in the following numerical experiments.

\subsubsection{2-d harmonic oscillator}
Subsequently, two-dimensional simulations with $r=5$ and $\bar{N}=1,2,3,4,5$ are considered, with initial value
\begin{equation}
	\psi(x_1,x_2)
	=\psi_1(x_1)\psi_2(x_2).
\end{equation}
Firstly, the integral of density and the density errors w.r.t. $\bar{N}=5$ are presented in Tab. \ref{tab:err-2d harmonic}. Next, the densities with and without bath are demonstrated in Fig. \ref{fig:den-2d harmonic}, followed by the density contributions from specific $\bar{N}$ in Fig. \ref{fig:dif-2d harmonic}. To clearly illustrate the validity of the proposed method, one-dimensional slices using the $x_2=0$ and $x_1=0$ planes are shown in Figs. \ref{fig:den-x slide-2d harmonic} and \ref{fig:den-y slide-2d harmonic}, respectively.
\begin{table}[H]
    \centering
    \caption{Integral of density (top)/ $L^2$ error of density difference (bottom) w.r.t. the one with $\bar{N}=5$.}
    \label{tab:err-2d harmonic}
	\begin{tabular}{c|ccccc}
		\hline
        \rule{0pt}{1em}
		$\bar{N}$	&1	&2	&3	&4	&5
		\\\hline
		$\int\rho$
		&0.9482    &0.9476    &0.9492    &0.9486    &0.9488
		\\
		$\|\rho-\rho^{\rm ref}\|_2$
		&0.0667    &0.0236    &0.0067    &0.0025 	&
		\\\hline
	\end{tabular}
\end{table}
As shown in Tab. \ref{tab:err-2d harmonic}, the density integral is close to the anticipated value of 1. Convergence to the reference density is also obtained.
\begin{figure}[H]
    \centering
    \includegraphics[width=.4\linewidth]{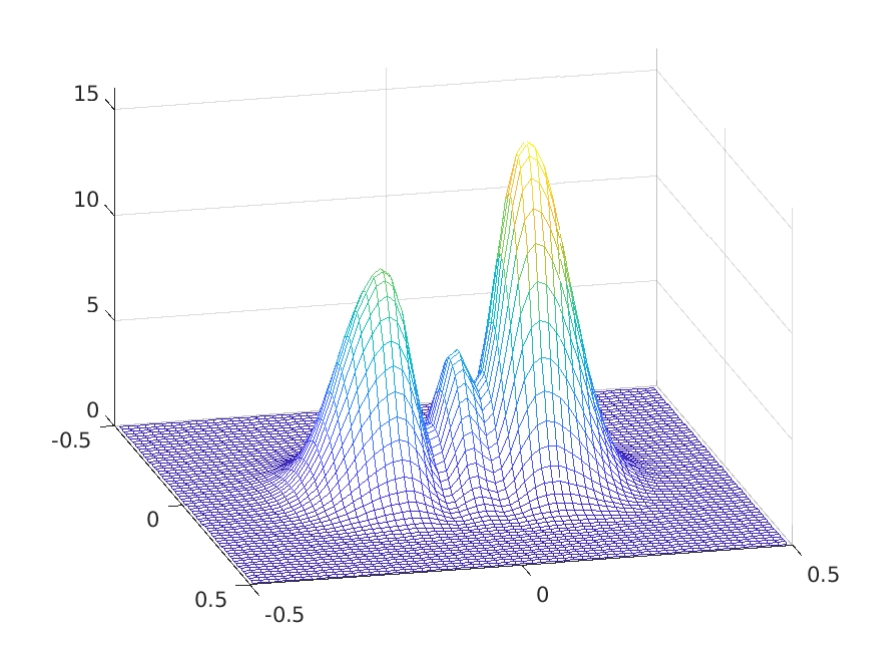}
    \includegraphics[width=.4\linewidth]{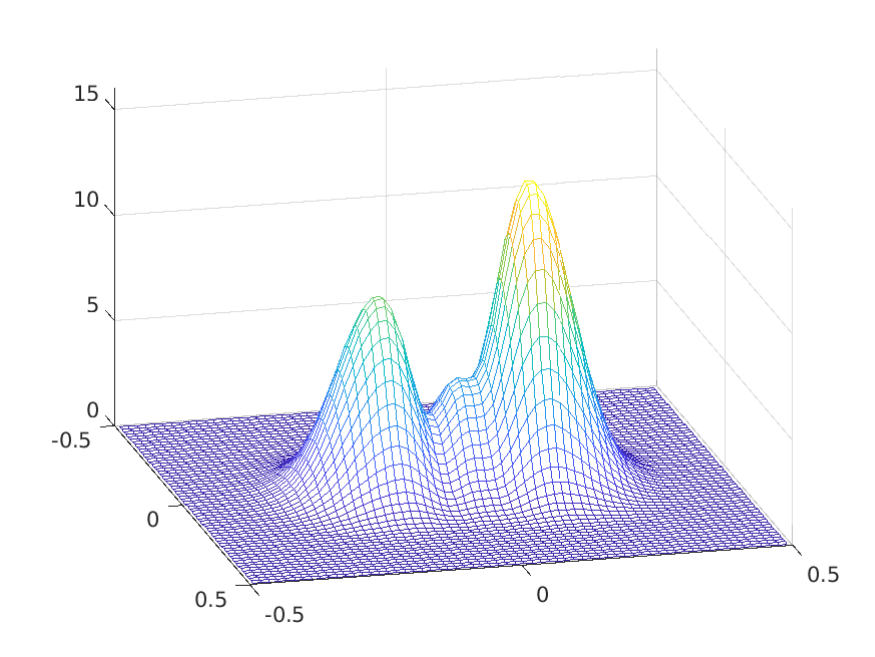}
    \caption{Density without bath (left)/ when $\bar{N}=5$ (right).}
    \label{fig:den-2d harmonic}
\end{figure}
In Fig. \ref{fig:den-2d harmonic}, the density shows the anticipated profile, which is a product between the one-dimensional result with initial value $\psi_1$ in \cite{wang2025solving}, and the one-dimensional Gaussian resulting from the contribution of $\psi_2$. However, small differences can be observed in the above figures, which we will take a closer look at in Figs. \ref{fig:den-x slide-2d harmonic} and \ref{fig:den-y slide-2d harmonic}.
\begin{figure}[H]
    \centering
    \includegraphics[width=.24\linewidth]{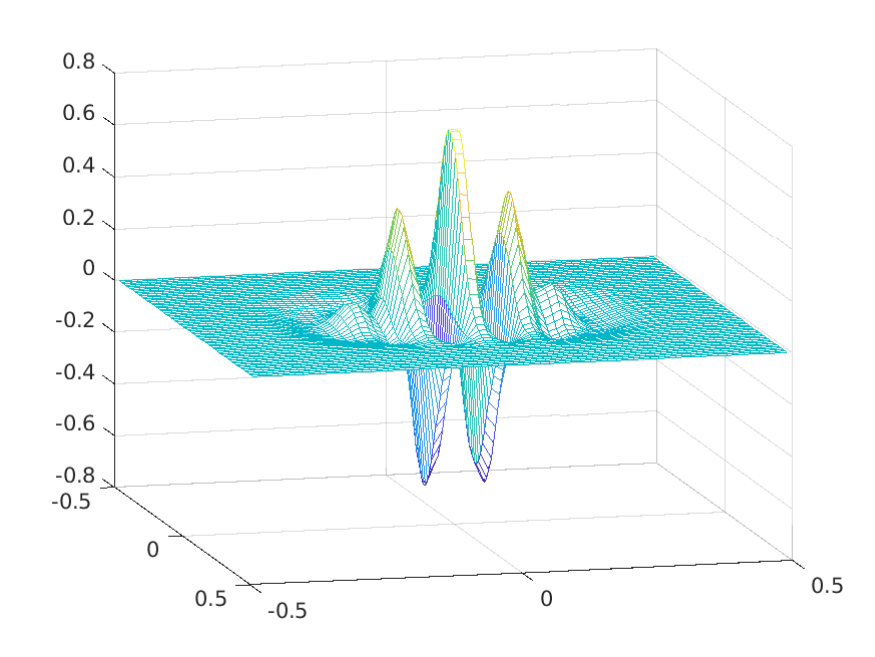}
    \includegraphics[width=.24\linewidth]{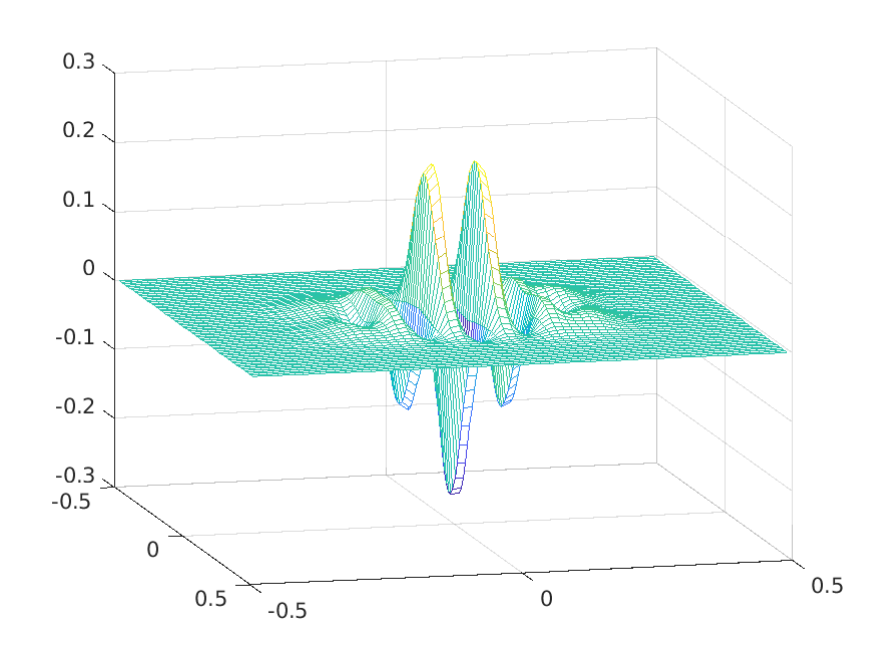}
    \includegraphics[width=.24\linewidth]{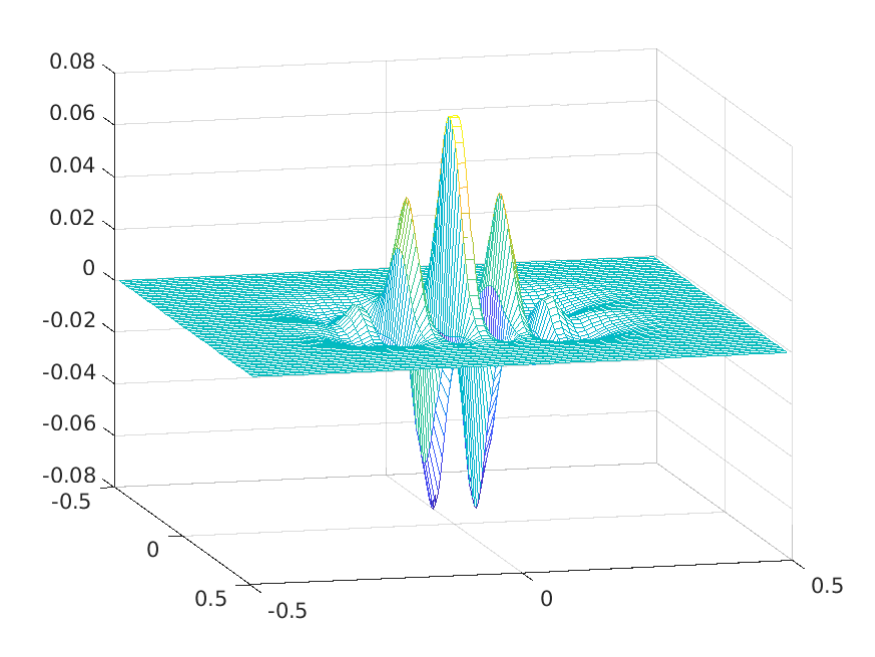}
    \includegraphics[width=.24\linewidth]{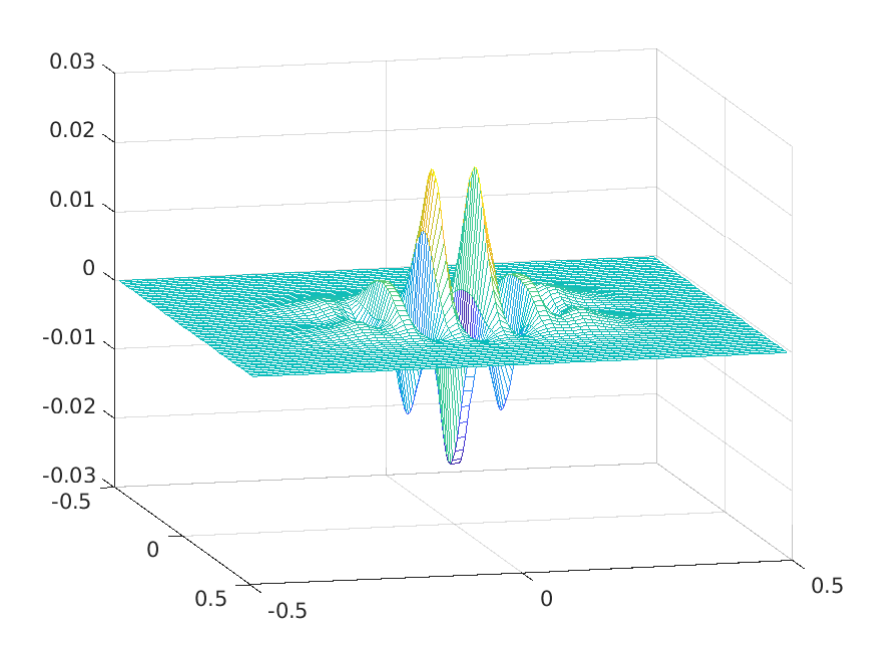}
    \caption{Density difference w.r.t. $\bar{N}=5$ density ($\bar{N}=1,2,3,4$ from left to right).}
    \label{fig:dif-2d harmonic}
\end{figure}
In Fig. \ref{fig:dif-2d harmonic}, convergence of the Dyson series can be obtained in the sense of a decrease in the axis range. We can also observe the change of sign in the contribution, corresponding to the mostly negative (positive) contributions in the 1st, 3rd (2nd, 4th) figures above.
\begin{figure}[H]
	\centering
	\includegraphics[width=.49\linewidth]{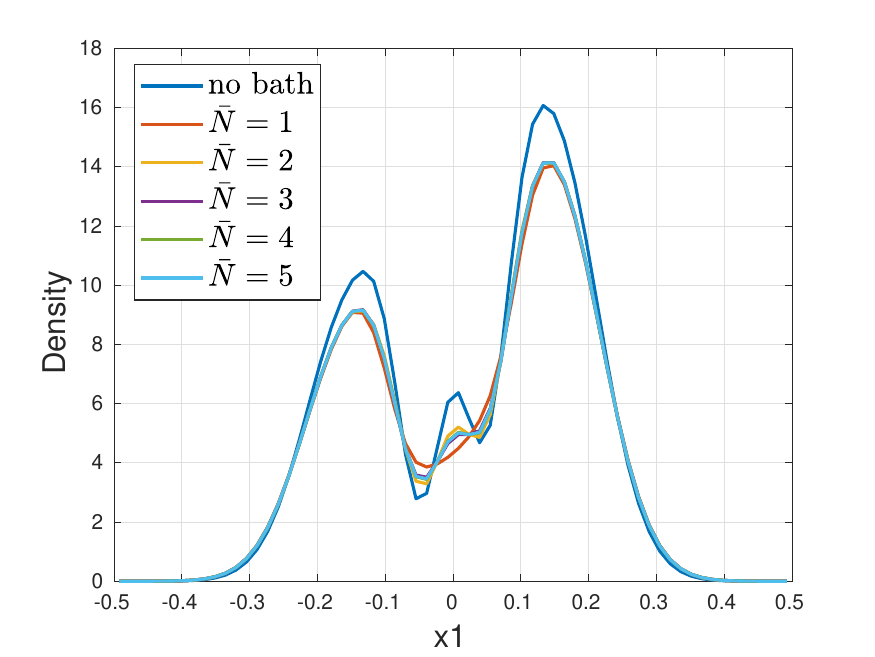}
	\includegraphics[width=.49\linewidth]{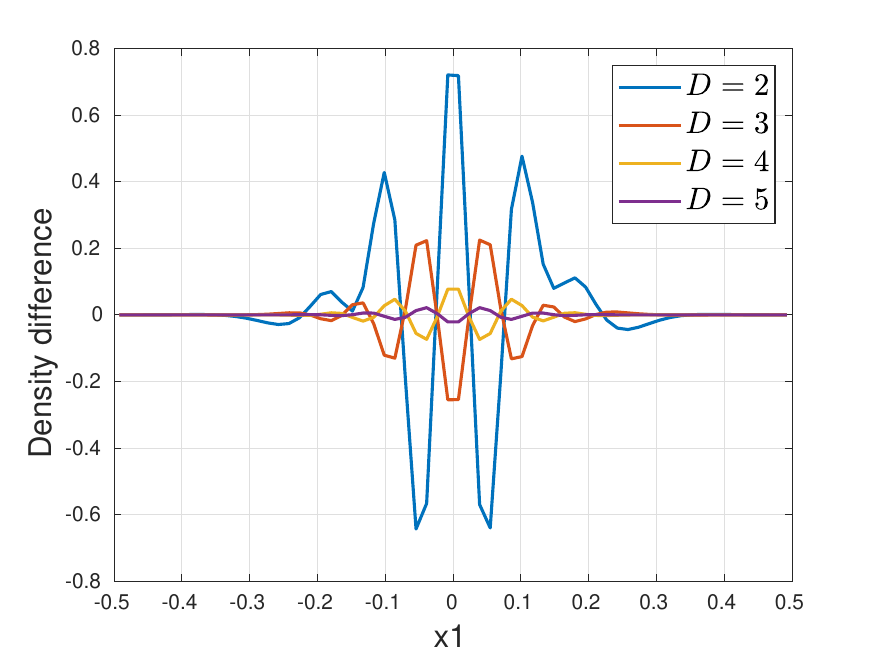}
	\caption{Density (left)/ Density difference (right) on $x_2 = 0$ plane.}
    \label{fig:den-x slide-2d harmonic}
\end{figure}
\begin{figure}[H]
	\centering
	\includegraphics[width=.49\linewidth]{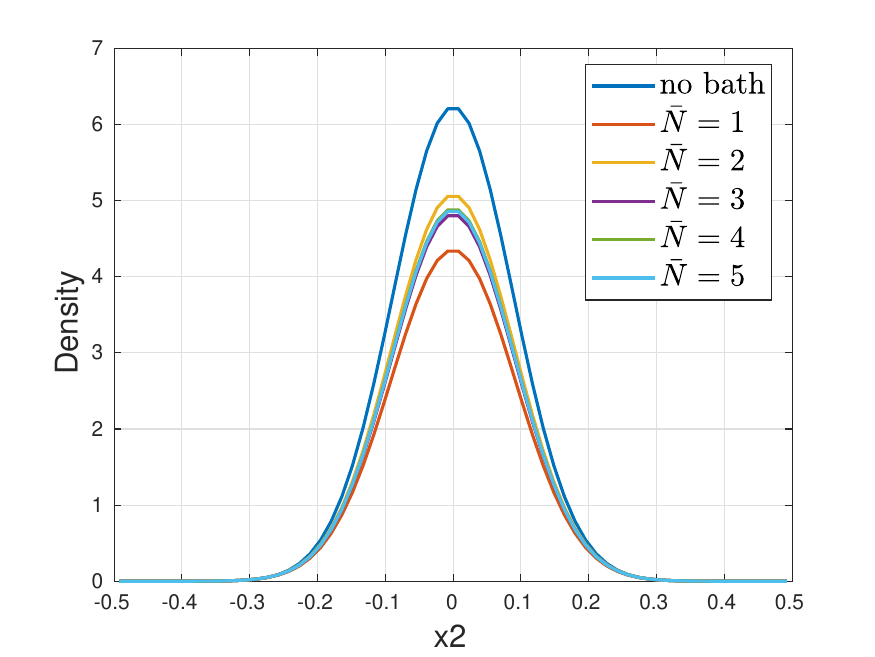}
	\includegraphics[width=.49\linewidth]{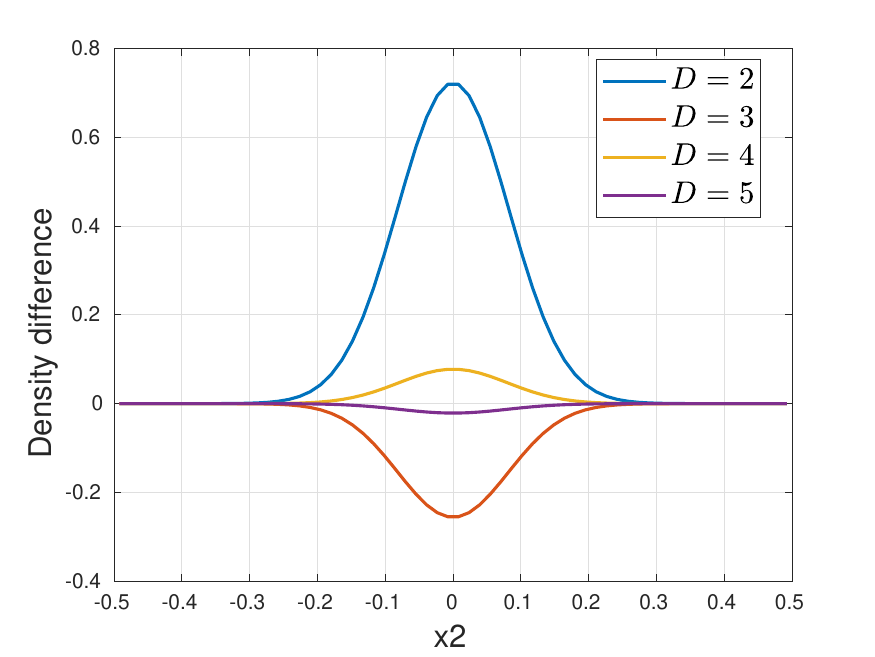}
	\caption{Density (left)/ Density difference (right) on $x_1 = 0$ plane.}
    \label{fig:den-y slide-2d harmonic}
\end{figure}
Figs. \ref{fig:den-x slide-2d harmonic} and \ref{fig:den-y slide-2d harmonic} show that (i) compared to the density without bath, the simulations with bath show a smoother behavior near the origin, due to quantum decoherence; (ii) convergence to the $\bar{N}=5$ result can be successfully observed in the figures on the right.

\subsection{Double slit}
Finally, we simulate the double slit interference experiment.
Fig. \ref{fig:den&pot} illustrates a visual representation of the double slit potential, as well as the computed initial and final densities. Due to the large initial momentum, the reflected part of the wave function is not observed in the final density.
\begin{figure}[H]
    \centering
    \includegraphics[width=0.49\linewidth]{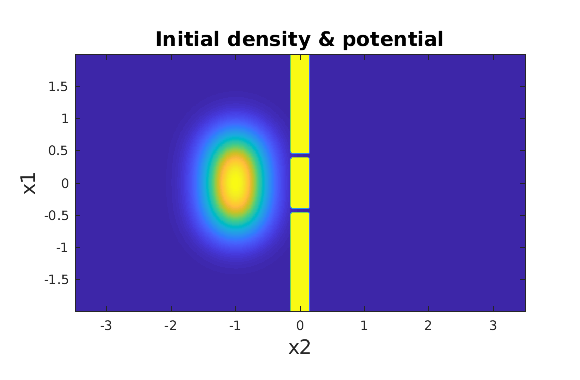}
    \includegraphics[width=0.49\linewidth]{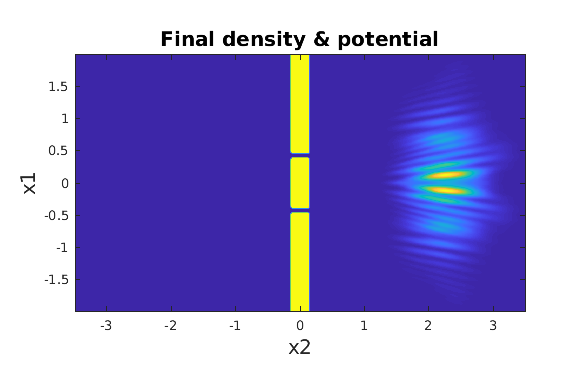}
    \caption{Initial (left)/ Final (right) density and double slit potential.}
    \label{fig:den&pot}
\end{figure}
Specifically, we choose the initial wave function
\begin{equation}
\begin{aligned}
    \psi_0(x_1,x_2)
    &=C_{\rm nor}^{\rm ds}\left(
        \exp\left(-\frac{(x_1-q_1)^2}{8\epsilon}\right)
        +\exp\left(-\frac{(x_1+q_1)^2}{8\epsilon}\right)\right)\\
        &\qquad\exp\left(-\frac{(x_2+q_2)^2}{8\epsilon}\right)
        \exp\left(i\frac{p_1x_1+p_2x_2}{4\epsilon}\right),
\end{aligned}
\end{equation}
with the normalization constant $C_{\rm nor}^{\rm ds}=(8\pi\epsilon(1+\exp(-q_1^2/(4\epsilon)))^{-1/2}$. To obtain the yellow shape in Fig. \ref{fig:den&pot}, the double slit potential is taken as the product of one-dimensional splines,
\begin{equation}
    V_{\rm ds}(x_1,x_2)
    =hV_{\rm ds}^{(1)}(x_1)V_{\rm ds}^{(2)}(x_2),
\end{equation}
where $h$ represents the height of the barrier,
\begin{equation}
    V_{\rm ds}^{(1)}(x_1)
    =\left\{\begin{array}{cc}
        1,
            &\text{if }|x_1|<d_1,\\
        f\left(\frac{d_1+b-|x_1|}{b}\right),
            &\text{if }d_1\le|x_1|<d_1+b,\\
        0,
            &\text{if }d_1+b\le|x_1|<d_1+b+w,\\
        f\left(\frac{|x_1|-d_1-b-w}{b}\right),
            &\text{if }d_1+b+w\le|x_1|<d_1+2b+w,\\
        1,
            &\text{if }d_1+2b+w\le|x_1|,
    \end{array}\right.
\end{equation}
and
\begin{equation}
    V_{\rm ds}^{(2)}(x_2)
    =\left\{\begin{array}{cc}
        1,
            &\text{if }|x_2|<d_2,\\
        f\left(\frac{d_2+b-|x_2|}{b}\right),
            &\text{if }d_2\le|x_2|<d_2+b,\\
        0,
            &\text{if }d_2+b\le|x_2|.
    \end{array}\right.
\end{equation}
The functions $V_{\rm ds}^{(1)}$ and $V_{\rm ds}^{(2)}$ are symmetric, with origin at the center of the barrier. From the origin, $d_1$ is the distance to a slit and $d_2$ is the distance to the edge of the barrier. $w$ is the width of each slit and $b$ is the width of the buffer between the top and bottom of the barrier. To ensure smoothness, the interpolating function $f$ is a six-order polynomial satisfying
\begin{equation*}
    f(0)=0,\qquad
    f(1)=1,\qquad
    f'(0)=f'(1)=f''(0)=f''(1)=0.
\end{equation*}
We set the parameters $h=10$, $d_1=0.35$, $d_2=0.1$, $w=0.05$, $b=0.05$, and $q_1=d_1+b+0.5w$, $q_2=-1$, $p_1=0$, $p_2=8$. We consider $\epsilon=1/16$ and $\b{p}\in[-1.5,1.5]\times[-1,1], \b{q}\in[-2,2]\times[-1,1]$, with stepsizes $\Delta p=\Delta q=1/16$. The simulation time is $t=0.4$ with stepsize $\Delta t=2.5\times10^{-4}$. Computations are performed for $\bar{N}=0,1,\ldots,5$.

In the following, we validate the choices of $r$ for the low-rank approximation and $\Delta t$ in Tab. \ref{tab:L2dif-Neig} and Tab. \ref{tab:L2dif-dt}, respectively. Subsequently, we show the densities with and without bath in Fig. \ref{fig:den with(out) bath}, and plot the density differences between adjacent $\bar{N}$ in Fig. \ref{fig:den dif-N}.

For a closer look at the bath influence, we take a one-dimensional slice at $x_2=2.3$ and compare the densities and density differences in Fig. \ref{fig:den/dif slides-N}. Lastly, the densities of slices with different coupling strength are overlaid in Fig. \ref{fig:den slides-xi}.
\begin{table}[H]
    \centering
    \caption{$L^2$ difference between adjacent $r$ when $\bar{N}=5$.}
    \begin{tabular}{c|cccc}\hline
        $r$ 
            &1  &2  &3  &4  \\\hline
        $\|\rho^{r}-\rho^{r+1}\|_2$
            &2.8274e-02 &2.4153e-07 &8.7320e-10 &8.7320e-10\\\hline
    \end{tabular}
    \label{tab:L2dif-Neig}
\end{table}
\begin{table}[H]
    \centering
    \caption{$L^2$ difference between adjacent $\Delta t$.}
    \begin{tabular}{c|cccc}\hline
        $\Delta t$ 
            &1.0e-03  &5.0e-04  &2.5e-04 &1.25e-04\\\hline
        $\|\rho^{\Delta t}-\rho^{\Delta t/2}\|_2$
            &1.2106e-01 &7.0600e-03 &1.0050e-03 &1.5145e-04\\\hline
    \end{tabular}
    \label{tab:L2dif-dt}
\end{table}
Convergence can be observed in Tabs. \ref{tab:L2dif-Neig} and \ref{tab:L2dif-dt}. As the computational cost is closely related to the low-rank approximation, we take $r=2$ and $\Delta t=2.5\times10^{-4}$ in the following experiments.
\begin{figure}[H]
    \centering
    \includegraphics[width=0.49\linewidth]{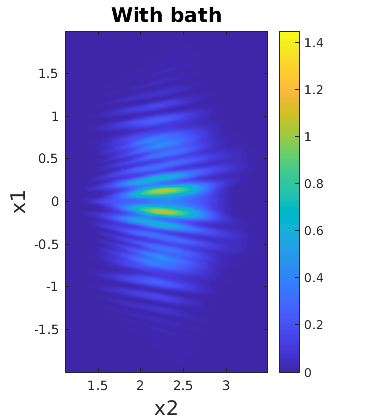}
    \includegraphics[width=0.49\linewidth]{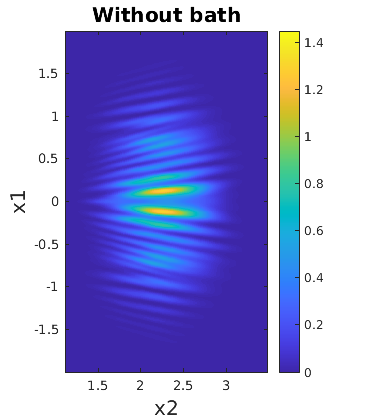}
    \caption{Density with bath (left)/ without bath (right).}
    \label{fig:den with(out) bath}
\end{figure}
\begin{figure}[H]
    \centering
    \includegraphics[width=.19\linewidth]{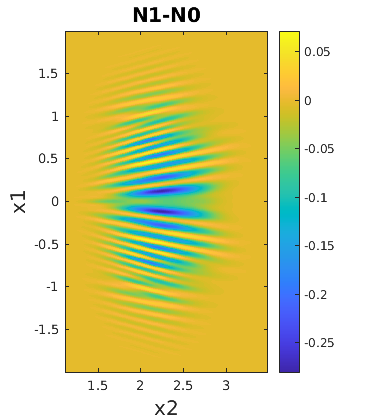}
    \includegraphics[width=.19\linewidth]{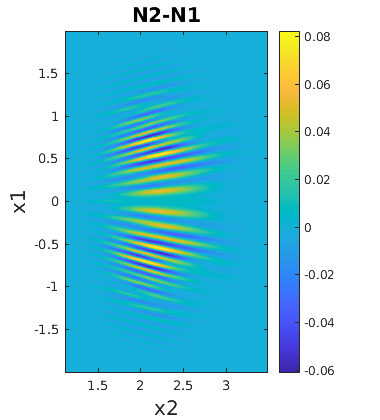}
    \includegraphics[width=.19\linewidth]{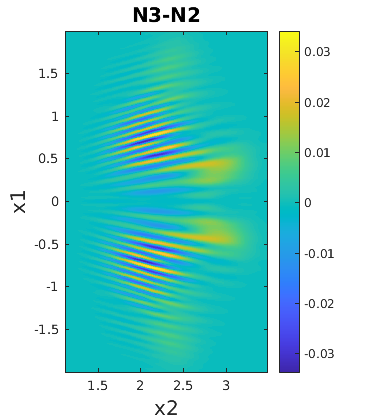}
    \includegraphics[width=.19\linewidth]{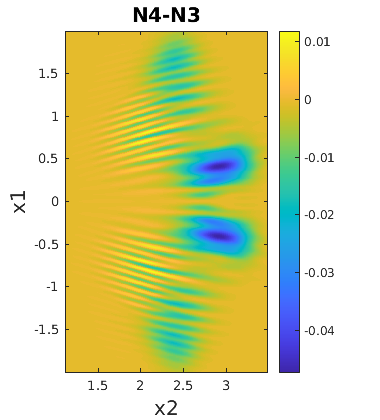}
    \includegraphics[width=.19\linewidth]{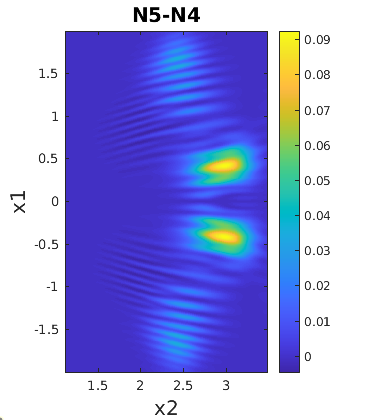}
    \caption{Density difference for different $\Bar{N}$.}
    \label{fig:den dif-N}
\end{figure}
Fig. \ref{fig:den with(out) bath} shows an oscillating density, resulting from quantum interference. In contrast, the simulation with bath delivers a smoother density, in the sense that different terms in the Dyson series, corresponding to figures in Fig. \ref{fig:den dif-N}, provide the effect of decoherence at oscillations. 
\begin{figure}[H]
    \centering
    \includegraphics[width=.49\linewidth]{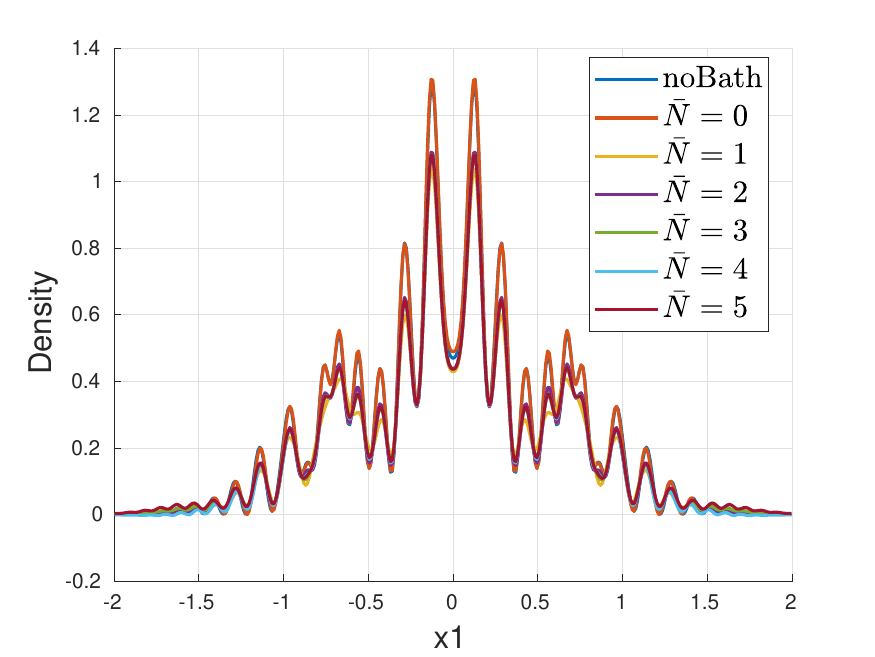}
    \includegraphics[width=.49\linewidth]{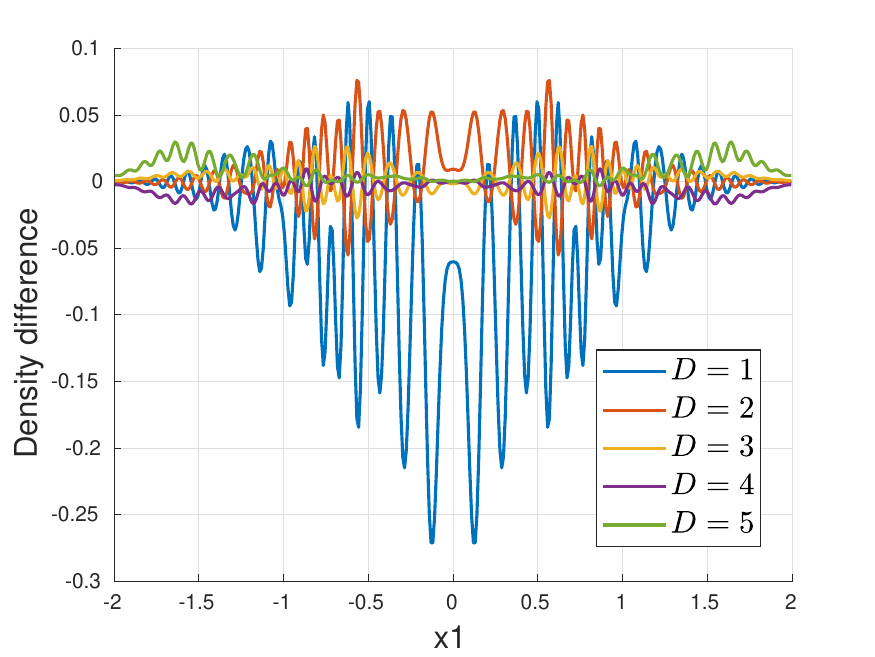}
    \caption{Density (left)/ Density difference (right) at $x_2=2.3$ for different $\bar{N}$.}
    \label{fig:den/dif slides-N}
\end{figure}
It can be observed in Fig. \ref{fig:den/dif slides-N} that (i) compared to without bath, simulations with bath exhibit results with fewer oscillations. (ii) As more terms in the Dyson series are considered ($\bar{N}$ increases), the oscillatory behavior decreases. (iii) The reduction in magnitude of the density difference demonstrates the convergence.
\begin{figure}[H]
    \centering
    \includegraphics[width=0.6\linewidth]{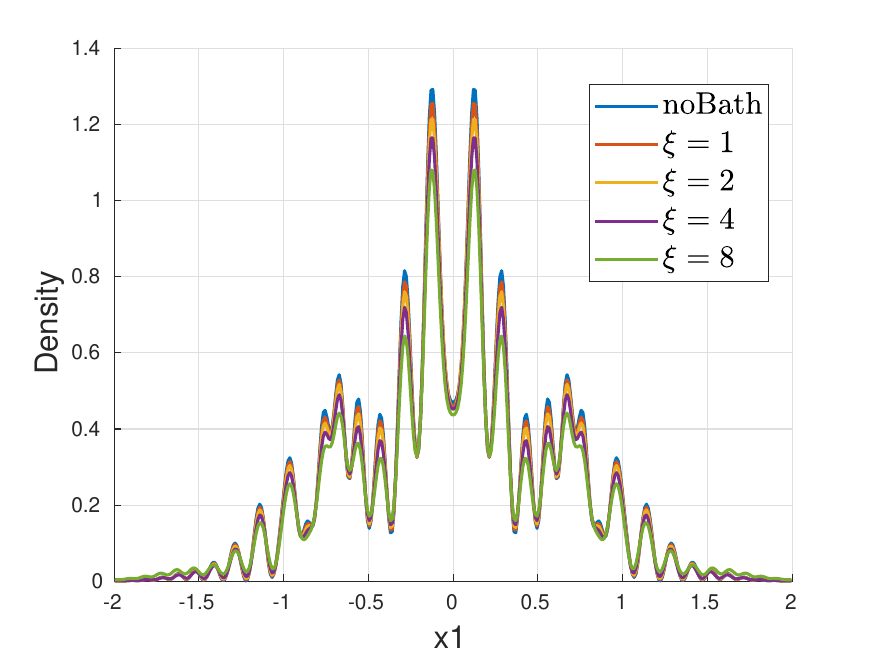}
    \caption{Density at $x_2=2.3$ for different $\xi$.}
    \label{fig:den slides-xi}
\end{figure}
In Fig. \ref{fig:den slides-xi}, simulations for a larger coupling constant $\xi$ exhibit smaller oscillations, due to the larger quantum effect of decoherence.
Although the suppression of oscillatory behavior into a single bright spot is expected under sufficiently strong quantum decoherence, the simulation with $\xi=8$ still exhibits a result similar to those without the bath. This can be primarily attributed to two factors: (i) the value of $\xi$ required for strong quantum decoherence exceeds the applicable range of the frozen Gaussian approximation; (ii) the effectiveness of the Dyson series expansion relies on treating the interaction operator as a perturbation, an assumption that may break down when the quantum decoherence is excessively strong.

\section{Conclusion}\label{sec:conc}

In this work, we reformulate the Dyson series for open quantum systems by separating the contributions in the reduced density matrix. By applying a low-rank approximation of the two-point correlation function, a complete separation of contributions along each axis is achieved, halving the spatial dimensionality. Furthermore, factorization of diagrams provides an additional reduction in computational cost.

Focusing on the multidimensional Caldeira–Leggett model, the frozen Gaussian approximation is used to lower the temporal dimensionality. These approximations enable the interaction operators to be represented as commutable scalars, thereby reducing high-dimensional integrals to one- and two-dimensional forms. Consequently, the dependence on time discretization is reduced to the same as that of the first non-trivial term in the Dyson series.
Through these techniques, we develop an efficient algorithm whose effectiveness is validated by several numerical experiments, including a two-dimensional double-slit setup.

In future works, we aim to further enhance the efficiency of the method to enable simulations of the three-dimensional Caldeira-Leggett model. 
One possible direction is the adoption of a higher-order frozen Gaussian approximation. 
Alternatively, we can address the computational bottleneck associated with the discretization of FGA. Direct discretization of the integrals arising from FGA, combined with Monte Carlo methods to evaluate these high-dimensional integrals, may offer a promising solution.
\section*{Author Contributions}
The third author conceived and led the project and supervised the research. The first author developed the methodology and implemented the CPU code. The second author implemented the GPU version of the code. All authors contributed to the discussion of the results and to the writing and revision of the manuscript.

\section*{Acknowledgments}
The work of Zhenning Cai was supported by the Academic Research Fund of the Ministry of Education of Singapore grant A-8002392-00-00.
We would also like to thank Geshuo Wang for helpful discussions.

\bibliographystyle{quantum}
\bibliography{ref}

\appendix

\section{Derivation of \eqref{eqn:den mat-3 branch}}

Starting from \eqref{eqn:dyson rho-interaction}, we first expand $\mathcal{L}_b(\b{s}^{(1)},\b{s}^{(2)})$ to get
\begin{equation} \label{eqn:dyson rho with B}
\begin{aligned}
    \rho_{s,I}(t)
    &=\sum\limits_{n_1,n_2=0}^\infty
        \int_{\b{s}^{(1)} \in \mathcal{S}_t^{n_1}} \int_{\b{s}^{(2)} \in \mathcal{S}_t^{n_2}}
        \Big((-i)^{n_1} G_s(\b{s}^{(1)})\Big)
        |\psi_s^{(0)}\rangle\langle\psi_s^{(0)}| \\
    &\qquad\qquad\qquad\qquad\qquad\qquad \cdot\Big((-i)^{n_2} G_s(\b{s}^{(2)})\Big)^\dagger
        \sum_{P \in \mathscr{P}_{n_1+n_2}} \prod_{(i,j) \in P} B(\tilde{s}_j, \tilde{s}_i)
            d\b{s}^{(1)}d\b{s}^{(2)},
\end{aligned}
\end{equation}
where $\tilde{\b{s}} = (s_1^{(1)}, \ldots, s_{n_1}^{(1)}, s_{n_2}^{(2)}, \ldots, s_1^{(2)})$. For a fixed $P \in \mathscr{P}_{n_1+n_2}$, any $(i,j) \in P$ is either:
(a) $i \leq n_1, j > n_1$, (b) $i, j \leq n_1$, (c) $i, j > n_1$.
By grouping such cases into the partition $P_{\rm cross} \cup P_1 \cup P_2 = P$, respectively, we can separate the terms as follows.
\begin{equation} \label{eqn:B cross_same}
\begin{aligned}
    \prod_{(i,j) \in P} B(\tilde{s}_j, \tilde{s}_i)
    &= \left(\prod_{(i,j) \in P_{\rm cross}} B(\tilde{s}_j, \tilde{s}_i)\right)
    \left(\prod_{(i,j) \in P_1} B(\tilde{s}_j, \tilde{s}_i)\right)
    \left(\prod_{(i,j) \in P_2} B(\tilde{s}_j, \tilde{s}_i)\right) \\
    &= \left(\prod_{(i,j) \in P_{\rm cross}} B(s_{n_1 + n_2 + 1 - j}^{(2)}, s_i^{(1)})\right) \\
    &\qquad \cdot \left(\prod_{(i,j) \in P_1} B(s_j^{(1)}, s_i^{(1)})\right)
    \left(\prod_{(i,j) \in P_2} B(s_{n_1 + n_2 + 1 - j}^{(2)}, s_{n_1 + n_2 + 1 - i}^{(2)})\right), \\
\end{aligned}
\end{equation}
where the index $n_1 + n_2 + 1 - j$ comes from the reverse ordering of $\b{s}^{(2)}$ in $\tilde{\b{s}}$. Note that
\begin{equation}
    B(s_{n_1 + n_2 + 1 - j}^{(2)}, s_{n_1 + n_2 + 1 - i}^{(2)}) = B(s_{n_1 + n_2 + 1 - i}^{(2)}, s_{n_1 + n_2 + 1 - j}^{(2)})^*,
\end{equation}
with $n_1 + n_2 + 1 - j < n_1 + n_2 + 1 - i$.
Instead of a summation over all $P \in \mathscr{P}_{n_1+n_2}$ followed by the partition of $P$ into 3 sets, we can first sum over all choices of partitions, followed by a product of all possible pairings within each set in the partition. Together with \eqref{eqn:B cross_same}, this gives the decomposition
\begin{equation} \label{eqn:L cross_same}
\begin{aligned}
    \sum_{P \in \mathscr{P}_{n_1+n_2}} \prod_{(i,j) \in P} B(\tilde{s}_j, \tilde{s}_i)
        &= \sum_{n = 0}^{\min\{n_1,n_2\}}
        \sum_{1 \leq k_1^{(1)} < \cdots < k_n^{(1)} \leq n_1 \atop 1 \leq k_1^{(2)} < \cdots < k_n^{(2)} \leq n_2}
        &&\left(\sum\limits_{\sigma\in\mathscr{Q}_n} \prod_{i=1}^n
        B(s_{k_{\sigma(i)}^{(2)}}^{(2)}, s_{k_i^{(1)}}^{(1)}) \right) \\
        &&&\cdot \left(\sum_{P_1 \in \mathscr{P}_{n_1-n}} \prod_{(i,j) \in P_1}
        B(\tau_j^{(1)}, \tau_i^{(1)}) \right) \\
        &&&\cdot \left(\sum_{P_2 \in \mathscr{P}_{n_2-n}} \prod_{(i,j) \in P_2}
        B(\tau_j^{(2)}, \tau_i^{(2)})^* \right),
\end{aligned}
\end{equation}
where $\boldsymbol{\tau}^{(j)} = (\tau_1^{(j)} , \ldots, \tau_{n_j - n}^{(j)})$ is defined as the subsequence $(s_i^{(j)})_{i=1}^{n_j} - (s_{k_i^{(j)}}^{(j)})_{i=1}^n$.

We define $m_1 = n_1 - n$ and $m_2 = n_2 - n$ to get
\begin{equation} \label{eqn:n cross_same}
    \sum\limits_{n_1,n_2=0}^\infty \sum_{n = 0}^{\min\{n_1,n_2\}}
    = \sum\limits_{n=0}^\infty
    {\color{blue}\sum\limits_{m_1=0}^\infty}
    {\color{red}\sum\limits_{m_2=0}^\infty}.
\end{equation}
Using the sum over all possible subsequences of $\b{s}^{(1)}$ and $\b{s}^{(2)}$ with length $n$, we can reformulate the integrals in \eqref{eqn:dyson rho with B}, after which renaming $s_{k_i^{(j)}}^{(j)}$ back to $s_i^{(j)}$ gives
\begin{equation} \label{eqn:integral cross_same}
\begin{aligned}
    &\int_{\b{s}^{(1)} \in \mathcal{S}_t^{n_1}} \int_{\b{s}^{(2)} \in \mathcal{S}_t^{n_2}}
        \sum_{1 \leq k_1^{(1)} < \cdots < k_n^{(1)} \leq n_1 \atop 1 \leq k_1^{(2)} < \cdots < k_n^{(2)} \leq n_2}
        (\cdot)
        d\b{s}^{(1)}d\b{s}^{(2)}\\
    &= \int_{\b{s}^{(1)}, \b{s}^{(2)} \in \mathcal{S}_t^n}
        {\color{blue}\int_{\boldsymbol{\tau}^{(1)} \in \mathcal{S}_t^{m_1}}}
        {\color{red}\int_{\boldsymbol{\tau}^{(2)} \in \mathcal{S}_t^{m_2}}}
        (\cdot) d\boldsymbol{\tau}^{(1)}d\boldsymbol{\tau}^{(2)}d\b{s}^{(1)}d\b{s}^{(2)}.
\end{aligned}
\end{equation}
Additionally, the time ordering operator $\mathcal{T}$ in the definition of $G_s$ allows us to use the notation
$G_s\left([\b{s}^{(j)},\boldsymbol{\tau}^{(j)}]\right)$ in the reformulated integral.

Finally, we apply \eqref{eqn:L cross_same} \eqref{eqn:n cross_same} \eqref{eqn:integral cross_same} to \eqref{eqn:dyson rho with B} to obtain the desired result \eqref{eqn:den mat-3 branch},
where the extra requirement for $m_j$ to be even comes from $\mathcal{L}_b^{\rm same}(\boldsymbol{\tau}^{(j)})$ being 0 when $m_j$ is odd.

\end{document}